\documentclass[final]{siamltex}
\usepackage{amsmath,amsfonts,amssymb}
\usepackage[usenames]{color}
\usepackage{bm}
\usepackage{braket}
\usepackage{xspace}
\usepackage{paralist}
\usepackage{graphicx}
\usepackage{epstopdf}
\usepackage{boxedminipage}
\usepackage{setspace}
\usepackage{tabularx}
\usepackage{multirow}
\usepackage{rotating}
\usepackage{subfig}
\usepackage{algorithm}
\usepackage{algorithmic}
\usepackage[colorlinks,urlcolor=blue,citecolor=blue,linkcolor=blue]{hyperref}

\newcommand{\cC}{{\cal C}}

\newcommand{\cP}{{\cal P}}

\newcommand{\RR}{\mathbb{R}}

\newcommand{\cei}[1]{\lceil #1 \rceil}

\newcommand{\eps}{\varepsilon}

\newcommand{\EX}{\hbox{\bf E}}

\newcommand{\Sec}[1]{\hyperref[sec:#1]{\S\ref*{sec:#1}}} 
\newcommand{\App}[1]{\hyperref[sec:#1]{Appendix~\ref*{sec:#1}}} 
\newcommand{\Eqn}[1]{\hyperref[eq:#1]{(\ref*{eq:#1})}} 
\newcommand{\Fig}[1]{\hyperref[fig:#1]{Fig.\,\ref*{fig:#1}}} 
\newcommand{\Tab}[1]{\hyperref[tab:#1]{Tab.\,\ref*{tab:#1}}} 
\newcommand{\Thm}[1]{\hyperref[thm:#1]{Thm.\,\ref*{thm:#1}}} 
\newcommand{\Lem}[1]{\hyperref[lem:#1]{Lem.\,\ref*{lem:#1}}} 
\newcommand{\Prop}[1]{\hyperref[prop:#1]{Prop.~\ref*{prop:#1}}} 
\newcommand{\Cor}[1]{\hyperref[cor:#1]{Cor.~\ref*{cor:#1}}} 
\newcommand{\Def}[1]{\hyperref[def:#1]{Defn.~\ref*{def:#1}}} 
\newcommand{\Ex}[1]{\hyperref[ex:#1]{Ex.~\ref*{ex:#1}}} 
\newcommand{\Clm}[1]{\hyperref[clm:#1]{Claim~\ref*{clm:#1}}} 
\newcommand{\Step}[1]{\hyperref[step:#1]{Step~\ref*{step:#1}}} 

\newcommand{\cema}{\cC}

\newcommand{\pmetric}{P}
\newcommand{\pt}{{p}}
\newcommand{\samp}{k}
\newcommand{\thresh}{\tau}

\newcommand{\be}{\begin{equation}}
\newcommand{\ee}{\end{equation}}
\newcommand{\bea}{\begin{eqnarray}}
\newcommand{\eea}{\end{eqnarray}}

\begin{document}

\title{Trigger detection for adaptive scientific workflows using percentile sampling \thanks{This work was funded by the Laboratory Directed Research and Development (LDRD) program of  Sandia National Laboratories. Sandia National Laboratories is a multi-program laboratory managed and operated by Sandia Corporation, a wholly owned subsidiary of Lockheed Martin Corporation, for the U.S. Department of Energy's National Nuclear Security Administration under contract DE-AC04-94AL85000.}}

\author{Janine C. Bennett\footnotemark[1], Ankit Bhagatwala\footnotemark[1], Jacqueline H. Chen\footnotemark[1], C. Seshadhri\footnotemark[2], Ali Pinar\footnotemark[1], Maher Salloum\footnotemark[1] }

\maketitle

\renewcommand{\thefootnote}{\fnsymbol{footnote}}
\footnotetext[1]{Sandia National Laboratories, Livermore, CA.
  Email: \{jcbenne, abhagat, jhcehn, apinar, mnsallo\}@sandia.gov}
\footnotetext[1]{ Dept. Computer Science, Univeristy of California at Santa Cruz, 
  Email: scomandu@ucsc.edu}
\renewcommand{\thefootnote}{\arabic{footnote}}

\begin{abstract}
Increasing complexity of both scientific simulations and high
performance computing system architectures are driving the need for 
\emph{adaptive workflows}, in which the composition and execution of  
computational and data manipulation steps dynamically depend on 
the evolutionary state of the simulation itself. 
Consider for example, the frequency of data storage.  Critical phases of the
simulation should be captured with high frequency and with high fidelity 
for post-analysis, however we cannot afford to retain the same 
frequency for the full simulation due to the high cost of data movement.  
We can instead look for triggers, indicators  that the simulation will be 
entering a critical phase, and adapt the workflow accordingly. 

In this paper, we present a methodology for detecting triggers and demonstrate
its use in the context of direct numerical simulations of turbulent combustion using S3D. We show 
that chemical explosive mode analysis (CEMA) can be  used to devise a 
noise-tolerant indicator for rapid increase in heat release. However, 
exhaustive computation of CEMA values dominates  the total simulation, 
thus is prohibitively expensive.  To overcome this computational 
bottleneck, we propose a quantile sampling approach.  Our sampling 
based algorithm  comes with provable error/confidence bounds, as a 
function of the number of samples. Most importantly,  the number of 
samples is independent of the problem size, thus our proposed sampling 
algorithm offers perfect scalability.  Our experiments on homogeneous charge 
compression ignition (HCCI)  and reactivity controlled compression ignition (RCCI) simulations 
show that the proposed method can  detect rapid increases in heat release, and its 
computational overhead is negligible. Our results will be used to make dynamic workflow 
decisions regarding data storage and mesh resolution in future combustion simulations.
The proposed sampling-based framework is
generalizable and we detail how it could be applied to a broad class of
scientific simulation workflows.  
\end{abstract}


\section{Introduction}
Steady improvements in computing resources enable ever more enhanced scientific
simulations, however Input/Output (I/O) constraints are impeding their impact.
Historically, scientific computing workflows have been defined by three 
independent stages (see \Fig{workflow}(a)):  
1) a pre-processing stage comprising 
initialization and set up (for example mesh generation, or
initial small-scale test runs); 2) the scientific computation itself (in which data
is periodically saved to disk at a prescribed frequency); and 3)
post-processing and analysis of data for scientific insights.
With improved computing resources scientists are increasing the temporal resolution of their 
simulations. However, as computational power continues to outpace I/O capabilities, 
the gap between time steps saved to disk keeps increasing. 
This compromise in the fidelity of the data being 
saved to disk makes it impossible to track features with 
timescales smaller than that of I/O frequency. Moreover, this situation is projected to worsen 
as we look ahead to future architectures with improvements in computational power continuing
to significantly outpace I/O capabilities~\cite{doe_arch,dav_exascale}, see
\Tab{exascale}. 

\begin{figure}[tbp] 
   \centering
   \includegraphics[width=5in]{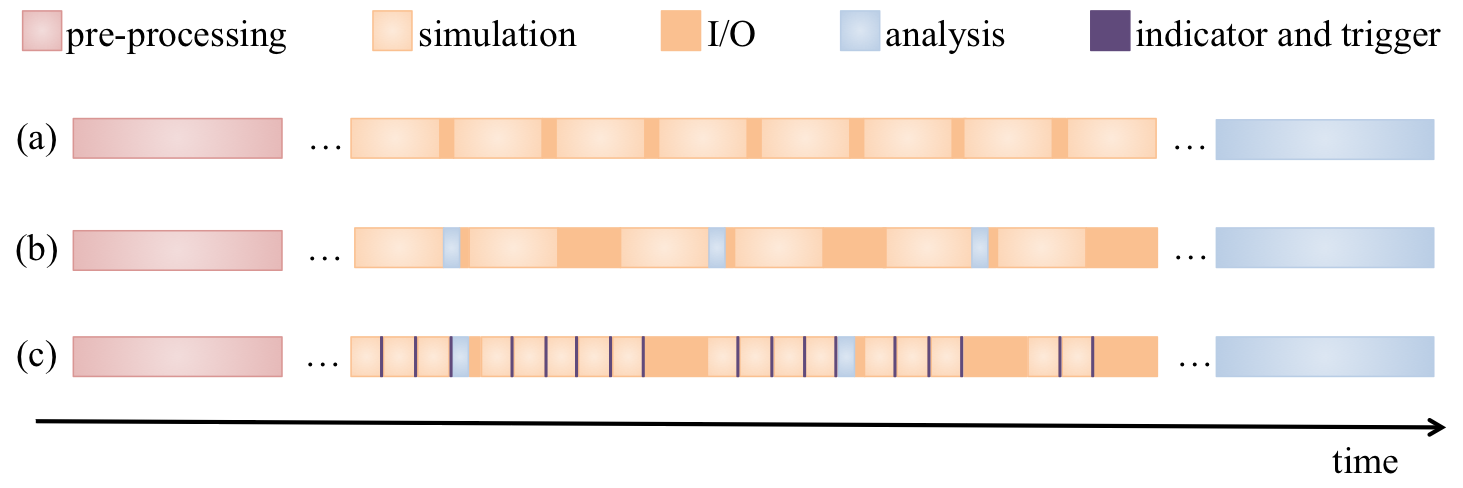}\\
   \caption{(a) An illustration of a traditional workflow made up of 3 stages:
   1) pre-procesing, 2) scientific computation and I/O at a prescribed rate, 3)
 analysis as a post-process. (b) Computational capabilities are outpacing I/O
 on future architectures, causing a change in workflows as some portion of the
 analysis moves \emph{in situ}. Most current day workflows remain static, with the frequency of I/O and
 analysis being prescribed upfront by the scientists. (c) This paper introduces
 the use of indicators and triggers to support adaptive workflows.  The
 indicator and trigger are lightweight functions evaluated at a high  frequency
 to make dynamic data-driven control-flow decisions.}
   \label{fig:workflow}
\end{figure}

Consequently, we are seeing a paradigm shift away from the use of prescribed I/O frequencies and 
post-process-centric data analysis, towards a more flexible concurrent 
paradigm in which raw simulation data is processed \emph{in situ} as it is computed, see
\Fig{workflow} (b).  In spite of the paradigm shift, concurrent processing 
does not provide a complete solution, as it requires all analysis questions be posed 
\emph{a priori}. This is not always possible as 
scientists often analyze their data in an interactive and exploratory fashion.
One potential solution, is to store data judiciously for only the time segments 
that will merit further analysis. However, the problem with this approach is
that the computation required to automatically and adaptively make decisions
regarding the workflow (e.g,. I/O and/or \emph{in situ} data analysis frequencies)
based on simulation state can be prohibitively expensive in their
own right. Therefore, we argue that one of the more pressing fundamental research challenges 
is the need for efficient, adaptive, data-driven control-flow mechanisms for
extreme-scale scientific simulation workflows,  which is the primary goal of this work. 

\begin{table}[t]
\caption{\label{tab:exascale} Expected exascale architecture
parameters for the design of two ``swim lanes'' of very different 
design choices~\cite{doe_arch, dav_exascale}. Note the drastic   difference
between expected improvements in I/O  and compute capacities in both swim
lanes.}  
\centering
 \begin{tabular}{|c|c|c|c|c|c|} 
 \hline

 System Parameter & 2011 &\multicolumn{2} {|c|} {2018} & Factor Change \\ 
 \hline\hline
 \hline
 System Peak & 2 Pf/s & \multicolumn{2}{|c|}{1 Ef/s} & 500 \\
 \hline
 Power & 6 MW & \multicolumn{2}{|c|}{$\le20$ MW} & 3 \\
 \hline
 System Memory & 0.3 PB & \multicolumn{2}{|c|}{32-64 PB} & 100-200 \\
 \hline
 Total Concurrency & 225K & 1B$\times$ 10 & 1B $\times$ 100 & 40000-400000 \\
 \hline
 Node Performance & 125 GF & 1TF & 10 TF& 8-80 \\
 \hline
 Node Concurrency & 12 & 1000 & 10000 & 83-830 \\
 \hline
 Network Bandwidth & 1.5 GB/s & 100 GB/s & 1000 GB/s& 66-660 \\
 \hline
 System Size (nodes) & 18700 & 1000000 & 100000 & 50-500 \\
 \hline
 I/O Capacity & 15 PB & \multicolumn{2}{|c|}{30-100 PB} & 20-67 \\ \hline
 I/O Bandwidth & 0.2 TB/s & \multicolumn{2}{|c|}{20-60 TB/s} & 10-30 \\\hline
\end{tabular}
\end{table} 
 
Here we introduce a methodology that is broadly applicable, yet can
be specialized to provide confidence guarantees for application-specific simulations and
underlying phenomena. 
Our methodology comprises three steps:  
\begin{compactenum} 
\item  Identify a noise-resistant \emph{indicator} that can be used to track changes in
  simulation state. 
\item  Devise a \emph{trigger} which specifies that a property of the indicator has been met.  
 \item  Design efficient and scalable algorithms to compute indicators and triggers.
 \end{compactenum} 
To make decisions in a data-driven
fashion, a user-defined \emph{indicator} function must be computed and measured \emph{in situ} 
at a relatively regular and high-frequency.  Along with the indicator, the application scientist
defines an associated \emph{trigger}, a function that returns a boolean value
indicating that the indicator has met some property, for example a threshold
value.  Together, indicators and triggers define data-driven control-flow
mechanisms within a scientific workflow, see \Fig{workflow}(c).  While this 
methodology is very intuitive and conceptually quite simple,
the challenges lie in defining indicators and triggers that capture the appropriate
scientific information while remaining cost efficient in terms of runtime, memory footprint, and I/O
requirements so that they can be deployed at the high frequency that is
required.  

In this article we demonstrate how recent advances in sublinear
algorithms~\cite{FischerSurvey,RonSurvey, RubSurvey} can be used to create efficient
indicators and triggers to enable data-driven control-flow mechanisms,
even in those cases where standard implementations of the
indicator and trigger would be significantly too expensive.  Sublinear algorithms are designed to estimate 
properties of a function over a massive discrete domain while 1) accessing 
only a tiny fraction of the domain, and 2) quantifying the error or uncertainty 
due to using only a sample of the data.  Sublinear indicators and triggers
operate on a sample whose size is dependent on the accuracy
of the desired result, rather than the input size. Consequently, sublinear
indicators and triggers can be deployed with \emph{high confidence} to 
make workflow decisions in extreme-scale simulations.  Sublinear
algorithms have their limitations;  in particular, they are not amenable for those control-flow
decisions that are based on anomaly detection.  However, they are well suited
to control-flow decisions regarding general trends in data, e.g., based on 
quantile plots of trends of computationally expensive quantities of interest. 

While our proposed approach is general, the first two steps can be made application
dependent, leveraging some knowledge of the underlying physics.  
In this paper we demonstrate our approach applied to dynamic workflow decisions 
in the context of direct numerical simulations of turbulent combustion using
S3D~\cite{chen09}.  In this use case, workflow decisions regarding both grid resolution 
and I/O frequency can be made based on the detection of a 
rapid increase in heat release.  This means our indicator should be a precursor to
heat release, and not the heat release itself.  What precedes the heat release?
 (see \Fig{heat_release}.) 
Recent studies have shown that chemical explosive mode analysis (CEMA) is a good lead indicator of 
heat release events~\cite{lu,shan}, and in this paper we show that CEMA can be used to 
devise a noise-tolerant indicator function and trigger.    CEMA is a point-wise
metric computed at each grid point.  Our analysis of the distribution of  CEMA
values  shows that the range of 
CEMA values covered by the top percentiles, compared to the full range of CEMA values, 
shrinks right before heat release and then expands  afterwards. This change in
distribution of quantiles is illustrated in \Fig{CEMA_indicator} and provides the
basis for our indicator and trigger. 

Now that we have an intuition for what can be a good lead indicator for heat release, 
the next step is to devise an indicator function and associated trigger that
quantify  the shrinking/expansion of the portion of top quantiles of CEMA values 
over the full range. The challenge here is the noise tolerance.  While the range of 
the CEMA values is defined  by the minimum and  the maximum, these  values are, by definition, 
outliers of the distribution,  and thus their adoption will lead to
noise-sensitive triggers.  Instead we use the quantiles of the distribution,  
which  remain stable among instances of the same distribution. For instance, we can 
replace the minimum with the 1-percentile and the maximum with the
98-percentile, still capturing the concept of range, but yielding a much more 
noise-tolerant trigger, as supported by experimental results (see \Sec{sample}). 

The final step is to design efficient and scalable algorithms for the
indicators and triggers we have devised. 
Our experiments show that our CEMA-based indicators and triggers work well in practice,  
however, exhaustive computation of CEMA values can dominate  the total
simulation computation time and are thus prohibitively expensive.  To overcome 
this computational bottleneck, we propose a sampling approach to estimate the quantiles.  
Our sampling based algorithm comes with provable error/confidence bounds  that are only 
a function of the number of samples.  With only 48K samples, the error will be less than \%1 
with confidence \%99.9, which in large-scale simulation runs leads to only a few 
samples per processor.  Most importantly, the number of samples is independent of 
the problem size, thus our proposed sampling algorithm offers perfect scalability.  
Our experiments on homogeneous charge compression ignition (HCCI)  and reactivity 
controlled compression ignition (RCCI) simulations show that  the proposed method 
can detect heat release, with negligible computational overhead. Moreover our results
will be used to make dynamic workflow decisions regarding data storage and mesh
resolution in future combustion simulations.  
  
The rest of the paper is organized as follows. \Sec{background}
provides the background for our work; we first review  the need  for and 
the state of adaptive workflows. Then we overview the mathematically rich 
field of sublinear algorithms and describe how this field can be 
instrumental in designing indicators and triggers for adaptive workflows.  
We conclude \Sec{background} with motivation of the combustion use case that is used
throughout our paper.
\Sec{indicator_trigger}
discusses the process of identifying a noise-resistant indicator and trigger
for  phase change in a simulation, and includes physics intuitions for how and
why CEMA can be used to construct an indicator for heat release for our
combustion use case. In \Sec{sample} we demonstrate how to compute 
the indicator and trigger efficiently using a sublinear approach and we put all the pieces together in \Sec{app} 
to demonstrate our technique on a full scale simulation. Finally, we
conclude with \Sec{conc}. 

\section{Background}
\label{sec:background} 
We begin this section  by reviewing recent work in enabling complex
scientific computing workflows. We then provide an overview of sublinear algorithms and 
discuss how these mathematical techniques can be deployed to make data driven
decisions \emph{in situ}. 
We conclude this section with a brief overview of our combustion use case.

\subsection{Adaptive data-driven workflows  and concurrent analysis frameworks}
\label{sec:enabling}
As we move to next generation architectures, scientists are moving away from traditional workflows in
which the simulation state is saved at prescribed frequencies for
post-processing analysis.  There are a number of concurrent analysis frameworks
available, wherein raw simulation output is processed as it is computed, decoupling the analysis from I/O. 
Both \emph{in situ} \cite{Yu2010,visit:2011, paraview:ldav11}  and 
\emph{in transit}  \cite{glean:ldav11, JITStaging, Bennett:2012} processing are based on  
performing analyses as the simulation is running, storing only the results, which are typically  several orders 
of magnitude smaller than the raw data. This reduction mitigates the effects of limited disk bandwidth and
capacity. Operations sharing primary resources of the simulation are considered
\emph{in situ}, while \emph{in transit} 
processing involves asynchronous data transfers to secondary resources. 

Concurrent analyses are often performed at frequencies that are prescribed by the scientists \emph{a
priori}.  For those analyses that are not too expensive -- in terms of runtime
(with respect to a simulation time step), memory footprint, and output size -- 
the prescription of frequencies is a viable approach.  However, for those analyses 
that are too expensive, prescribed frequencies will not suffice because the 
scientific phenomenon that is being simulated typically does not behave linearly (e.g., combustion, climate,
astrophysics).  When scientists choose a prescribed I/O or analysis frequency 
that is frequent enough to capture the science of interest, the costs incurred are 
too great, while a prescribed frequency that is cost-effective and less
frequent may miss the underlying scientific effects that simulation is intended to capture.
An alternative approach would be to perform expensive analyses and I/O in an
adaptive fashion, driven by the data itself.  In ~\cite{nouanesengsy2014adr,
modelPaper}, such techniques have been developed based on entropy of
information in the data, and building piecewise-linear fits of quantities of
interest.  These approaches fit within the methodology proposed here and are domain-agnostic.
In this work we present a strategy that can leverage the scientists' physics
intuitions, even when the \emph{in situ} analyses that captures those
intuitions would otherwise be too expensive to compute. 

\subsection{Sublinear algorithms}
A recent development in theoretical computer science and mathematics is the
study of sublinear algorithms,  which aim to understand global features of a data set  while using limited resources. 
Often enough,  we do not need to look at the entire data to determine some of its important features.  The field of sublinear 
algorithms~\cite{FischerSurvey,RonSurvey, RubSurvey} makes precise the settings when this is possible and combines discrete math and algorithmic 
techniques with statistical tools to quantify error and give trade-offs with sample sizes.  This confidence measure is necessary for  adoption 
of such techniques by the scientific computing community, whose scientific results can be used
to make high-impact decisions.

Formally,  given a function, $f:D \rightarrow R$, we assume  for any $x \in D$, we can
query $f(x)$. 
For example, in S3D simulations we have  $D = [n]^3$ as a structured grid, and $R = \RR$ can be  the temperature values. 
If $D = [n]$ and $R = \{0,1\}$, then $f$ represents an $n$-bit binary string. If $R = \{A, T, G, C\}$,
 $f$ could represent a DNA segment. If $D = [n]^2$, the function could represent a matrix (or
a graph). Note $D$  can also be an unstructured grid, modeled as a graph. Similarly, almost any data analysis input can be cast as a 
collection of functions over a discrete or discretized domain.

We are interested in some \emph{specific} property of $f$, which is phrased  as
a yes/no question. For instance, in a jet simulation we can ask  if there  exists a high-temperature region spanning the $x$-axis of the grid. 
How can we determine if $f$ satisfies property $\cP$ \emph{without querying all of $f$}? 
It is impossible to give an exact answer without knowledge of $f$. To formalize what can be inferred
by querying $o(|D|)$ values of $f$, we use a notion of the \emph{distance to
$\cP$}~\footnote{ $f(n) = o(g(n))$ means for all $c > 0$ there exists some $k >
0$ such that $0 \leq f(n) < cg(n)$ for all $n \geq k$. The value of $k$ must not
depend on $n$, but may depend on $c$.}.
Every function $f$ has a distance to $\cP$, denoted $\eps_f$, where $\eps_f = 0$ iff $f$
satisfied $\cP$. To provide an exact answer to questions regarding $\cP$, we determine
whether $\eps_f = 0$ or $\eps_f \neq 0$. However, approximate answers can be given by 
choosing some error parameter $\eps > 0$ and then determining whether we  can distinguish
$\eps_f = 0$ from $\eps_f > \eps$.
The theory of sublinear algorithms shows whether the latter question can be resolved  by an algorithm that samples $o(|D|)$ function values.

For a sublinear algorithm, there are usually three parameters of interest: the number of samples $t$, the error $\eps$, and the
confidence $\delta$. As described earlier, the error is expressed as the distance to $\cP$.
Analysis shows that for a given $t$, we can estimate
the  answer within error $\eps$ with a confidence of $> 1 -\delta$. Conversely, given $\eps, \delta$,
we can compute the number of samples required.

Although at a high level, any question that can be framed in terms of determining global properties of a large domain is 
subject to a sublinear analysis, surprisingly, the origins of this field
have nothing to do with ``big data'' or computational challenges in data analysis. 
The birth of sublinear algorithms is in computational complexity theory~\cite{RS96}.
Hence, practical performances of  these methods on real applications have not been fully investigated. 
Recent work by some of the authors showcase the potential of sampling algorithms in graph
analysis~\cite{SePiKo13, SePiKo14, KoPiPlSe13, JhSePi13, JhSePi15,BaKo15}, and the generation of application-independent
generation of colormaps~\cite{TBSP13}.

\subsection{Potential of sublinear algorithms in enabling adaptive workflows}
The construction of a scalable and performant indicator function comprises a
number of technical challenges. First, algorithms will be sharing compute 
resources with the simulation, which poses constraints on the algorithmic choices. 
In particular, memory is expected to be a bottleneck in exascale computers 
and beyond (see \Tab{exascale}), and thus we may have to work with data 
structures of the application, and not be able to build auxiliary data 
structures that will improve the performance of our algorithms. Secondly, the 
data layout will be dictated by the simulation, locally at the node level 
and globally at the system level. This layout will not necessarily be favorable 
for our algorithms, yet pre-processing to move the data will be infeasible at 
large scales.  Thirdly,  computation of the indicator needs to be fast, so that it does not  
slow down  the simulation computation.  For instance,   our analysis for judicious I/O   
cannot take more time than the I/O itself.  

We expect sampling-based algorithms, especially sublinear algorithms, to play an 
important role in the design of indicator functions at the  exascale era and beyond.  First, 
the small number of samples grant runtime efficiency, which enable working 
concurrently with the simulation, with  negligible effect on runtime.  
The error/confidence bounds quantify the  compromise in accuracy compared to full analysis. Moreover, for most 
problems, the number of samples required only depend on  error/confidence bounds, 
which lead to perfect scalability of these algorithms for extreme problem sizes. The memory 
requirements of  sublinear algorithms are also  small, and typically  only in the 
order of the samples.  In some cases,  additional data structures may be necessary 
to enable random sampling,  but even  such structures are not memory-intensive. 

We claim that  sublinear algorithms can play a critical role  for \emph{in
situ} analysis. With this 
paper, we  will showcase one application of sublinear algorithms, and we hope that   our 
success will draw attention to this field with high potential. 

\subsection{Combustion use case}
Throughout this paper we demonstrate our approach applied to a combustion
use case, using S3D~\cite{chen09}, a direct numerical simulation (DNS) of 
combustion in turbulence.  
The combustion simulations in our use case pertain to a class of internal combustion
(IC) engine concept, called premixed-charge compression ignition (PCCI). The
central idea is that the air-fuel mixture is allowed to ignite on its own, as opposed to being forced to ignite through a spark, as
is done in conventional spark-ignited (SI) engines. This results in
substantial improvements in fuel efficiency. However, one of the roadblocks to
this technology is that the ignition is difficult to control. In
particular, it is difficult to predict the precise moment of ignition,
which is important for such an engine to be practical. It is undesirable to
have simultaneous ignition of the entire mixture, or even a large fraction
of the mixture, which would result in knocking, damaging the engine. \\ 

\begin{figure}[th]
\centering
\includegraphics[width=0.6\textwidth]{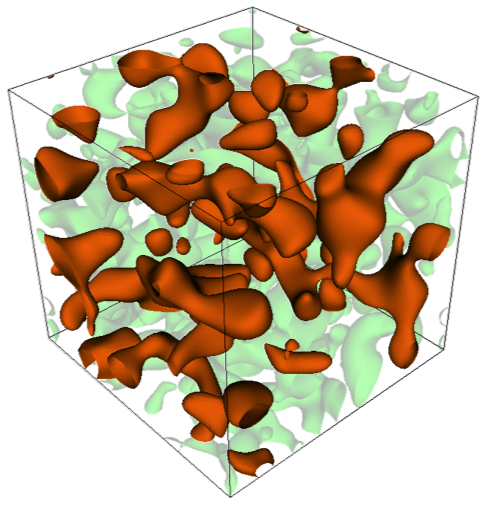}
\caption{\label{fig:kernels} In a Homogenous Compression Charge Ignition
simulation many small heat kernels slowly develop prior to auto-ignition.
In this image, regions of high heat kernels are shown in orange and regions
of high vorticity are shown in green.}
\end{figure}

Within the broad framework of PCCI, several ideas have been
proposed to alleviate the difficulty noted above. All of them seek to
control the ignition process by staggering it in time, i.e. different
parcels of the air-fuel mixture ignite at different times, so that the
overall heat release is delocalized in time. This is done typically by
stratifying the mixture, i.e. mixture properties are varied  spatially in
such a way that the desired heat release profile is obtained. Here, we are
interested in two specific techniques called homogeneous-charge compression
ignition (HCCI) \cite{bhagatwala1} and reactivity-controlled compression
ignition (RCCI) \cite{kokjohn,bhagatwala2}. In both cases, heat release
starts out in the form of small kernels at arbitrary locations in the
simulation domain, see \Fig{kernels}. Eventually, multiple kernels 
ignite as the overall heat release reaches a global maximum and subsequently 
declines. Since these simulations are computationally and storage-intensive, 
we want to run the simulation at a coarser grid resolution and save data less 
frequently during the early build-up phase.  When the heat release events 
occur,  we want to run the simulation at the finest grid granularity possible, 
and store the data as frequently as possible.  Therefore, it is imperative to 
be able to predict the start of the heat release event using an indicator and 
trigger that serve to inform the application to adjust its grid resolution and 
I/O frequency accordingly, see \Fig{heat_release}. \\

\begin{figure}[H]
\centering
\includegraphics[width=1.0\textwidth]{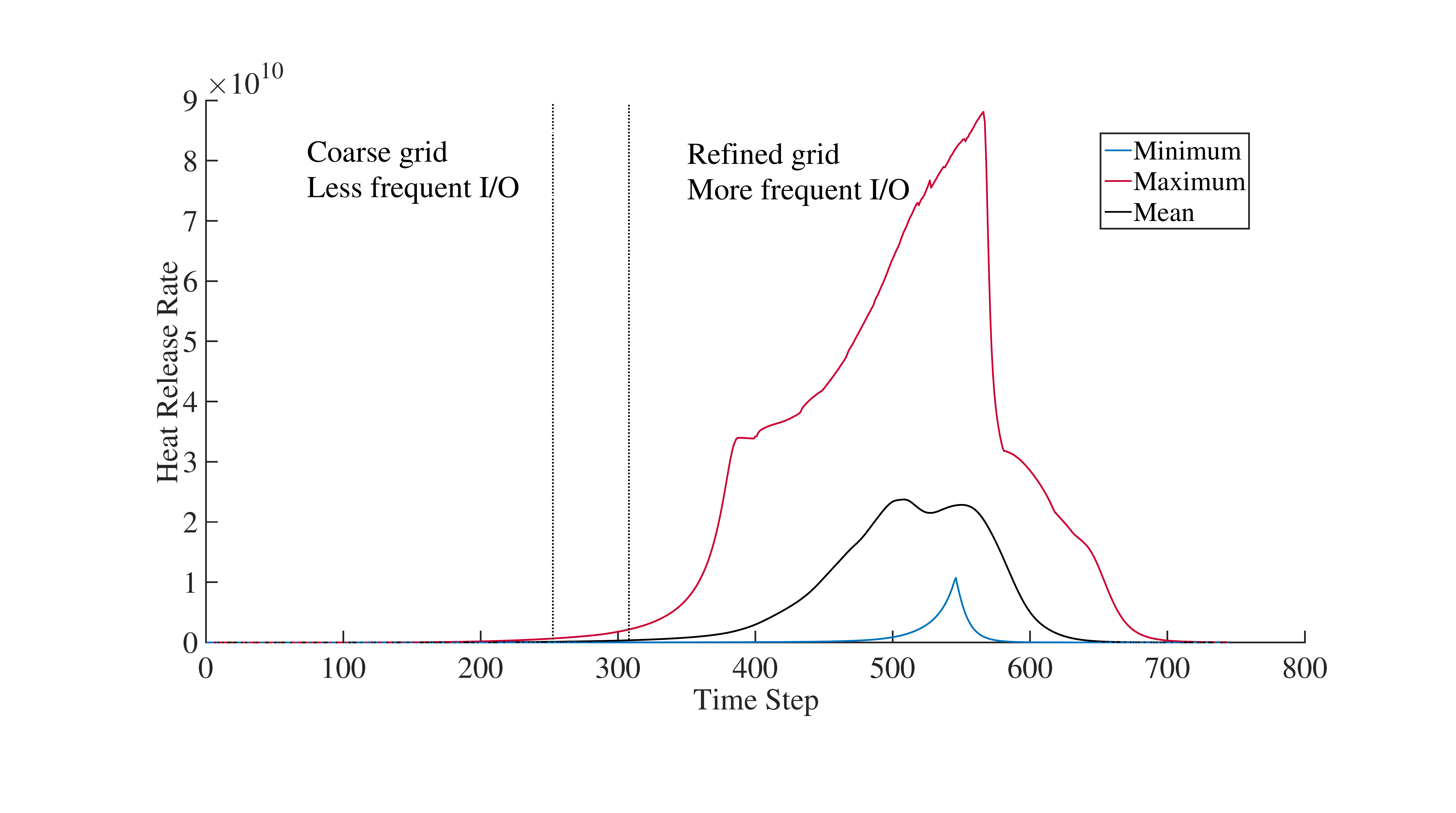}
\caption{\label{fig:heat_release} 
    The minimum (blue), maximum (red) and mean (black) heat release values
  for each time step in the simulation. Early in the simulation,
  we want to run at a coarser grid resolution and save data less frequently.
  When heat release events occur, we want to run the simulation at the finest
  granularity, and save the data as frequently as possible.  The vertical dotted lines
  in this figure define a range of time steps within which we would like to make
  this workflow transition (as identified by a domain expert).
}
\end{figure}

\section{Designing a Noise-Resistant Indicator and Trigger} 
\label{sec:indicator_trigger}
While general-purpose indicators could be computed (e.g. entropy of a quantity
of interest), we argue that application domain-specific indicators in many
cases will best capture the phenomena of interest.  In
this section we describe the design of an indicator and trigger for heat
release for our combustion use case.  
To provide context, we begin by discussing the intuitions that informed our
design.

\subsection{Chemical Explosive Mode Analysis}
\label{sec:intution}
One of the most reliable techniques to predict incipient heat release is
the chemical explosive mode analysis (CEMA).
CEMA is a pointwise computational technique described in detail by Lu \emph{et al.} \cite{lu} and Shan \emph{et al.}
\cite{shan}. A brief description is provided here for reference.
The conservation equations for reacting species can be written as 

\begin{equation*} 
\frac{D\mathbf{y}}{Dt} = \mathbf{g(y)} \equiv \mathbf{\omega(y)} + \mathbf{s(y)}
\end{equation*}

\noindent The vector $\mathbf{y}$ in CEMA represents temperature and reacting
species mass fractions, $\mathbf{\omega}$ is the reaction source term and
$\mathbf{s}$ is the mixing term. The Jacobian of the right hand side can be written as 

\[ 
\mathbf{J_g}  =  \frac{\partial \mathbf{g(y)}}{\partial \mathbf{y}} = \mathbf{J_\omega} + \mathbf{J_s}, \;\;{\mathrm where} \;\;
\mathbf{J_\omega} =  \frac{\partial \mathbf{\omega(y)}}{\partial \mathbf{y}}  \;\; {\mathrm and }\;\;
\mathbf{J_s}  =  \frac{\partial \mathbf{s(y)}}{\partial \mathbf{y}}. 
\]

\noindent The chemical Jacobian, $\mathbf{J_\omega}$ can be used to infer
chemical properties of the mixture. This is done using an
eigen-decomposition of the Jacobian. If the eigenvalues of the Jacobian
corresponding to the non-conservative modes are arranged in descending
order of the real part, $\lambda_e$ is defined as the first eigenvalue and
$\lambda_i$ are the remaining eigenvalues. The eigenmode associated with
$\lambda_e$ is defined as a chemical explosive mode (CEM) if 

\begin{equation}
\textrm{Re}(\lambda_e) > 0, \;\;{\mathrm for}\;\; \lambda_e = \mathbf{b_e J_\omega a_e},
\end{equation}   

\noindent where $\mathbf{b_e}$ and $\mathbf{a_e}$ are the left and right
eigenvectors respectively for $\lambda_e$. The presence of a CEM indicates
the propensity of a mixture to ignite. CEMA is a pointwise metric, typically
computed at every grid point in the simulation domain. The criterion
defined above then indicates whether that point will undergo ignition or
whether it has already undergone ignition. If it has undergone ignition, we
have $Re(\lambda_e)<0$.  

\begin{figure}[H]
\centering
\includegraphics[width=0.9\textwidth]{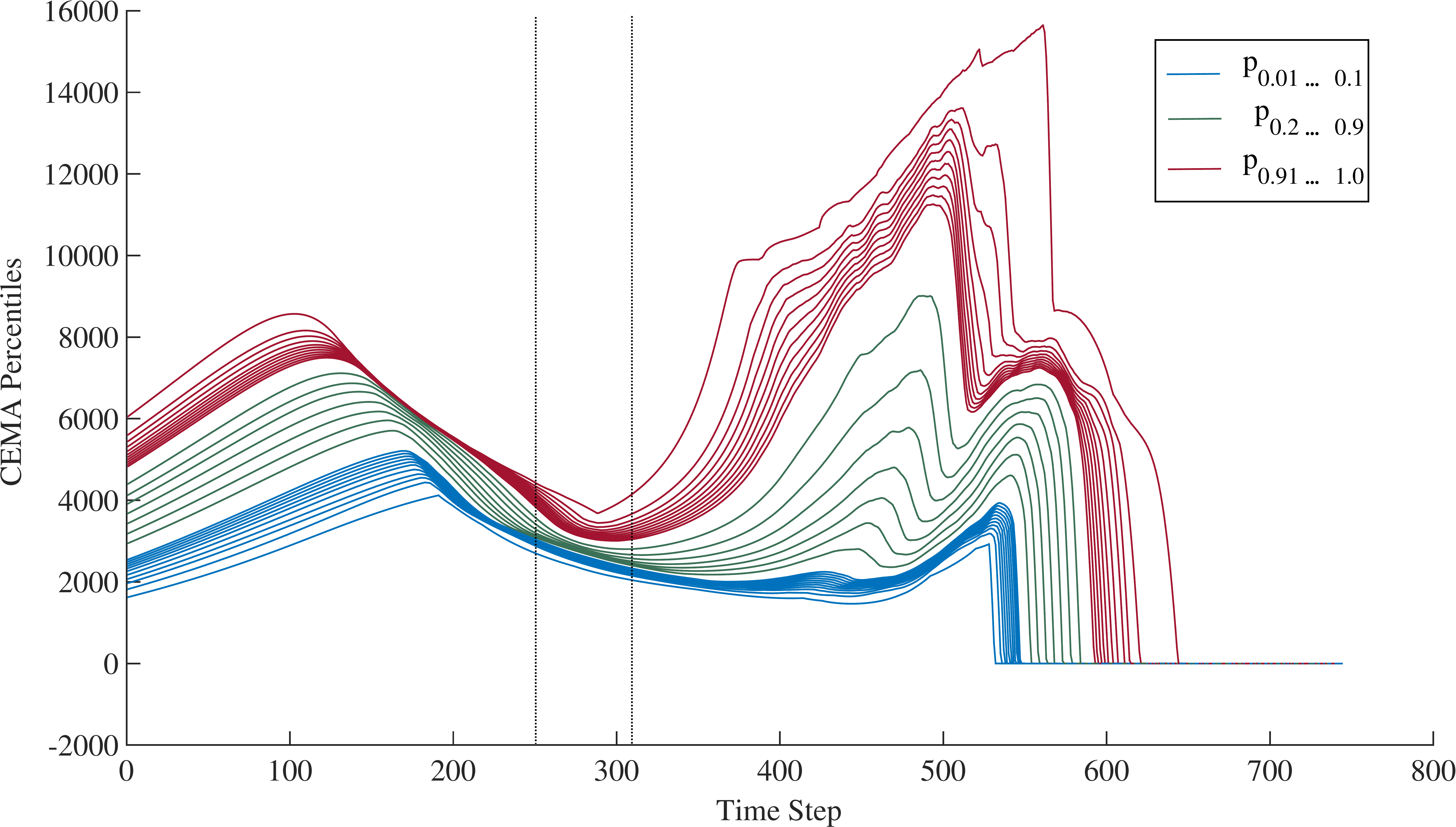}
\caption{\label{fig:CEMA_indicator} 
   In this percentile plot of CEMA values, the lowest
blue curve and the highest red curve correspond to the 1 and
100 percentiles ($p_{0.01}$ and $p_{1}$), respectively. The blue curves
correspond to $p_{0.01 \ldots 0.1}$, the green curves to $p_{0.2 \ldots 0.9}$
and the red curves to $p_{0.91 \ldots 1}$. We notice that as the simulation progresses, the distance between the higher percentiles (red
curves) decreases then suddenly increases. Our aim is to define a 
function that captures when this spread in the high percentiles occurs, as this
serves as a good indicator of ignition (it falls within the user-defined window
of true trigger time steps, indicated by the vertical dotted lines).
}
\end{figure}

Our CEMA-based indicator is based on global trends of CEMA over time.
Consider \Fig{CEMA_indicator} 
which provides a summary of the trends of CEMA values across all time steps in a simulation. 
At timestep $t$, let $\cema(t)$ be the array of CEMA values on the underlying mesh. 
It is convenient to think of $\cema(t) \in \RR^N$, where $N$ is the total number
of grid points. This array is distributed
among $M$ processors such that each process accesses $N/M$ points of the
field.  Let $\widehat{\cema(t)}$
be the sorted version of $\cema(t)$. For $\alpha \in (0,1]$, the $\alpha$-percentile 
is the entry $\widehat{\cema(t)}_{\lceil \alpha N \rceil}$. 
More specifically, it is the value in $\widehat{\cema(t)}$ that is greater than at least $\cei{\alpha N}$
values in $\widehat{\cema(t)}$. We denote this value by $\pt_\alpha(t)$.

We notice that as the simulation progresses, the distance between the higher percentiles (red
curves) decreases then suddenly increases. This is illustrated in the plot
by the spread in the red curves that occurs just before the dashed line (indicating
ignition). What does this mean? Suppose the range of CEMA values (at time $t$)
is $[x,y]$. If the distribution of CEMA values was uniform in $[x,y]$, then the 
$\alpha$-percentile would have value around $x + \alpha(y-x)$. 
Suppose the distribution was highly non-uniform, with a large fraction
of small (close to $x$) values. Then, for large $\alpha$, we expect
the $\alpha$-percentile to be smaller than $x + \alpha(y-x)$. 
Alternately, for (large) $\alpha < \beta$, in the uniform case,
the difference between these percentiles is $(\beta-\alpha)(y-x)$.
If many values are small, we expect this difference to be larger.
Essentially, the gaps between the high percentiles become larger
as more CEMA values move towards the lower end of the range. 
 
Our empirical observation is consistent with the underlying physics. From a physical 
point of view, this trend in the distribution of CEMA values indicates the formation of the first ignition kernels in the 
fuel-air mixture. As some of these kernels become fully burnt, their CEMA values 
become negative. As a parcel of fluid transitions from fully unburnt to partially 
burnt to fully burnt, its temperature increases monotonically and the CEMA value 
associated with it reaches a peak value before crossing zero and attaining a 
negative value indicating a fully burnt state. The CEMA values for several 
other kernels that are in different stages of ignition, i.e. partially burnt, 
lie between those for unburnt (large positive) and burnt (negative) mixtures. 
The large range of CEMA values in partially burnt mixtures explains the range 
of values seen in the percentile plots as ignition is initiated in the mixture.
\subsection{Designing a Noise-Resistant CEMA-Based Indicator} \label{sec:indicator}
We introduce an indicator function, \emph{\pmetric-indicator} that quantifies 
the distribution of the top quantiles of CEMA values over time. The
\pmetric-indicator measures the ratio of the range of the top percentiles  
to the full range of CEMA values. 
Our indicator function  relies on  the range of CEMA values,  which are 
defined by their minimum and maximum.  However, the maximum and the minimum of a 
distribution, by definition,  are outliers, and  thus they can change drastically  
between instances even when the underlying  distribution does not change.  Hence, 
we avoid the maximum and minimum and replace them with  high and low quantiles of 
the distribution. 

Formally, quantiles are defined as  values taken at regular intervals from the 
inverse of the cumulative distribution function of a random variable. For a given 
data set, quantiles are used to divide the data into equal sized sets  after sorting, 
and the quantiles are the values on the boundary between consecutive subsets.  A special 
case is dividing 100 equal groups, when we can refer to quantiles as
percentiles.  This paper focuses on percentiles with numbers in the $[0,1]$ 
range (although all techniques presented here can be generalized for any
quantiles). For example, the $0.5$ percentile will refer to the median of the data set. 
\begin{figure}[H]
\centering
\includegraphics[width=0.9\textwidth]{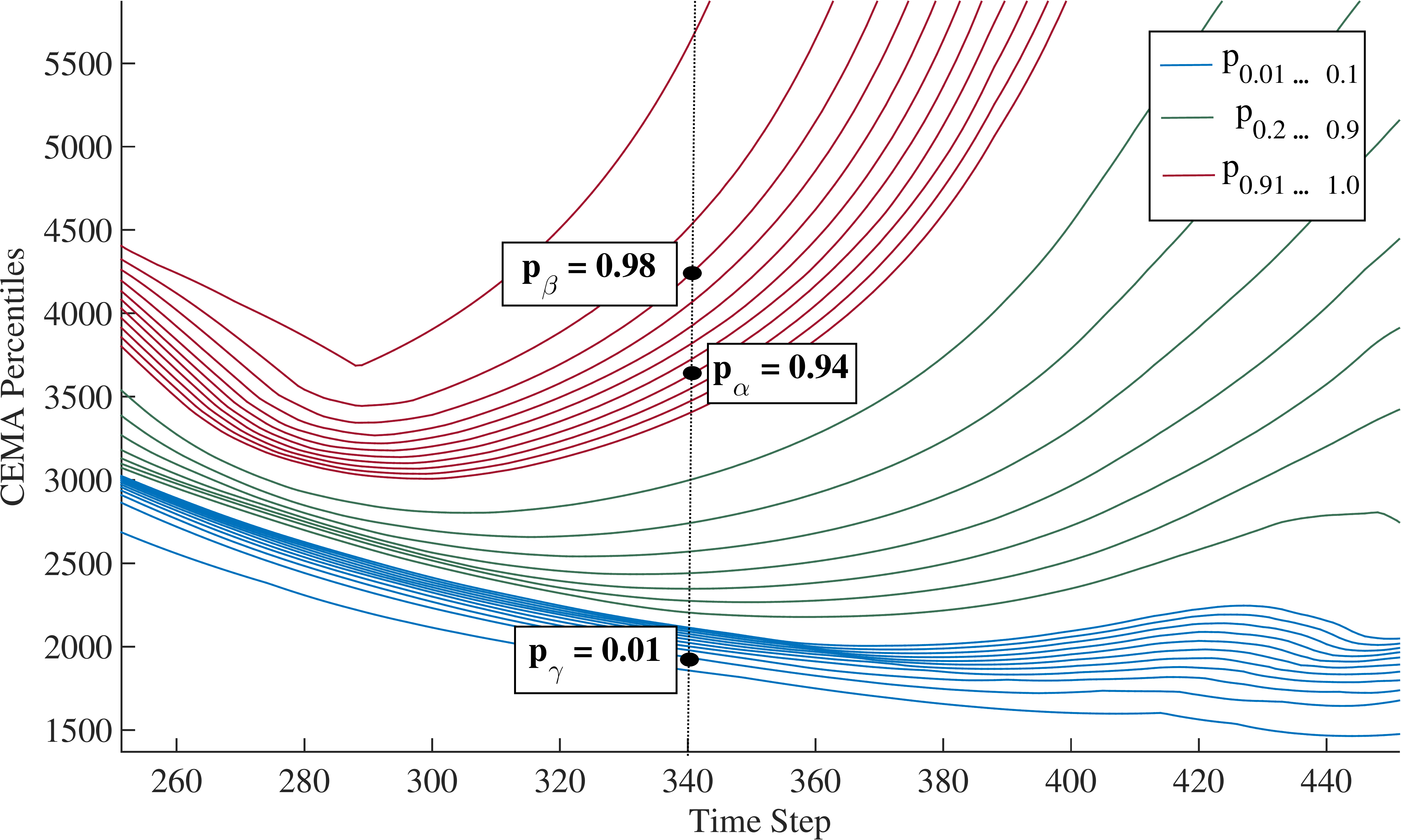}
\caption{\label{fig:p_indicator} 
  $\pmetric_{\alpha,\beta,\gamma}(t)$ values evaluated at each time step. 
  Here we illustrate $\alpha_0=0.94$, $\beta=0.98$ and $\gamma=0.01$ for
  time step 340 of a simulation.
}
\end{figure}
 
We substitute the maximum with a high percentile, which we denote by the 
$\beta$-percentile where $\beta $ is typically in the range $[0.95,0.99]$ and the 
minimum with a low percentile, which we denote by the $\gamma$-percentile where 
$\gamma$ is typically in the range $[0.01,0.05]$. This substitution provides  
stability to our measurements,  without compromising what we want to measure. 
 
Consider the following notation: Let $A \in \RR^N$, be an array in sorted order.
The $\alpha$ percentile of this sorted array is exactly the entry $A_{\lceil \alpha N \rceil}$. 
We use $\pt_\alpha$ to refer to this value (i.e.,  $\pt_\alpha = A_{\lceil \alpha N \rceil}$.  
The $A$ array  will change at each step of  the simulation, and thus we will use $A(t)$ and  
$\pt_\alpha(t)$ to refer to the data  on step $t$ of the simulation. 

We define our indicator on a given array $A$ using 3 parameters: $\alpha$,
$\beta$ and $\gamma$.  As described above, we use $\pt_\beta(t)$ and
$\pt_\gamma(t)$ as substitutes for the maximum and the minimum of $A(t)$
respectively, see \Fig{p_indicator}. We  want to detect whether the range
covered by top quantiles shrinks, and $\alpha$ represents the lower end of the
top quantiles. Therefore, the  range of top quantiles we  measure is
$\left[\pt_\alpha(t), \pt_\beta(t)\right]$.  In our indicator, we choose $\alpha
< \beta$ (typically in the range $[0.95,0.99]$) and $\gamma$ (typically in the
range $[0.01,0.05]$). We measure the spread at time $t$ by the
\pmetric-indicator:
\be
\pmetric_{\alpha,\beta,\gamma}(t) = \frac{\pt_\alpha(t) - \pt_\gamma(t)}{\pt_\beta(t) - \pt_\gamma(t)}.
\label{eq:pmetric}
\ee

\noindent In this indicator,  the denominator corresponds to the full range of CEMA values,  while  the 
numerator  corresponds to the  range  after the  top quantiles are removed.
 When the CEMA values are uniformly distributed, 
 $\pmetric_{\alpha,\beta,\gamma}(t) \approx (\alpha - \gamma)/(\beta - \gamma)$  $\approx \alpha/\beta$.
When there is a significant shift towards lower values, $\pmetric_{\alpha,\beta,\gamma}(t)$ will become smaller.

\begin{figure}[H]
\centering
\includegraphics[width=\textwidth]{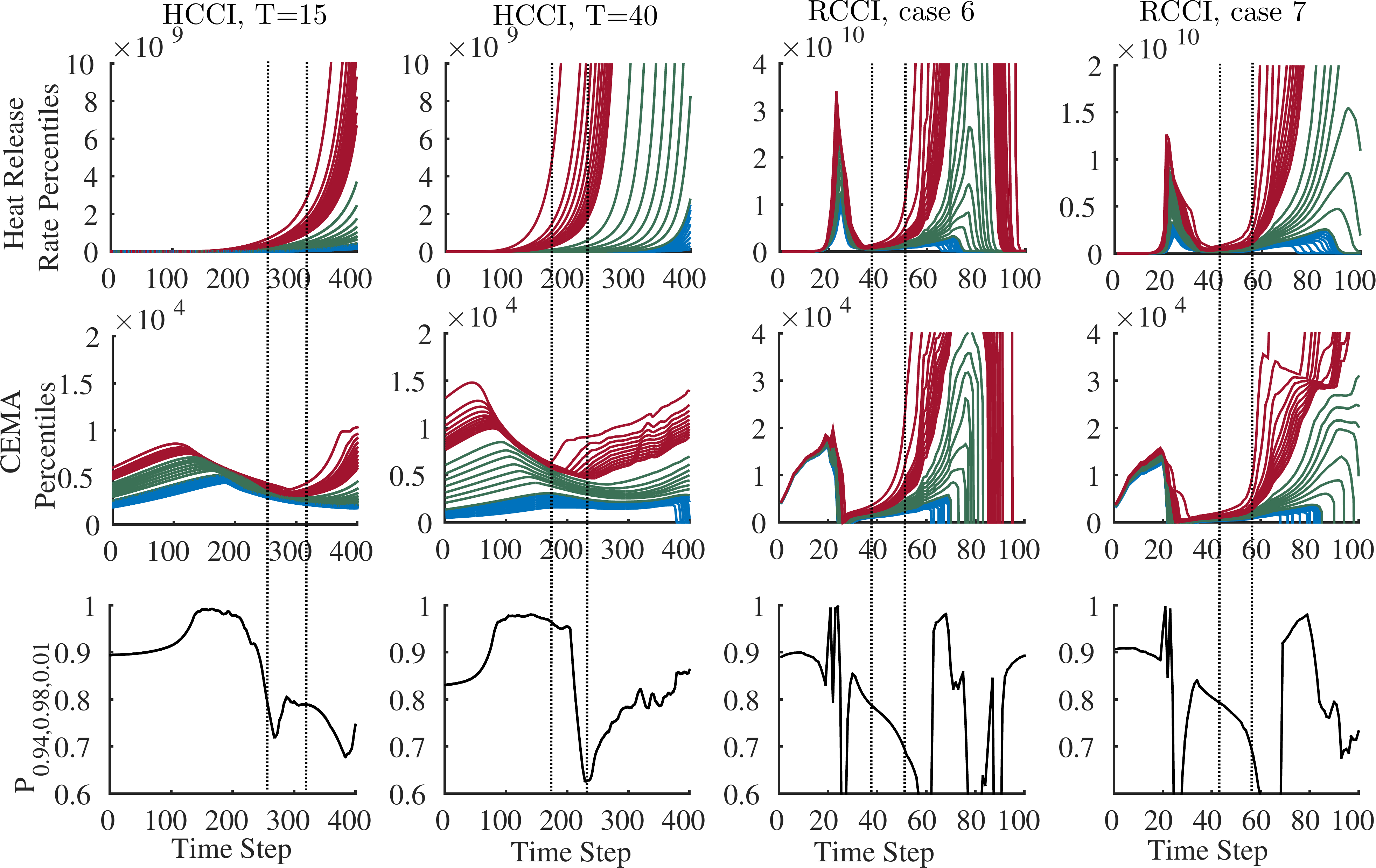}
\caption{\label{fig:CEMA} Plots showing the (top row) percentiles of the
heat release rate, (middle row) percentiles of CEMA, and (bottom row) the
\pmetric-indicator, as indicated. In the percentile plots, the lowest
blue curve and the highest red curve correspond to the 1 and
100 percentiles ($p_{0.01}$ and $p_{1}$), respectively. The blue curves
correspond to $p_{0.01 \ldots 0.1}$, the green curves to $p_{0.2 \ldots 0.9}$
and the red curves to $p_{0.91 \ldots 1}$. The \pmetric-indicator shown is
evaluated for $\alpha_0=0.94$, $\beta=0.98$ and $\gamma=0.01$.
The vertical dotted lines crossing the images indicate a window of
acceptable ``true'' trigger time steps, as identified by a domain expert.
For the RCCI cases, the trigger time ranges are based on the High Temperature
Heat Release (HTHR), i.e., the second peak in the Heat Release Rate (HRR) profiles. }
\end{figure}

\Fig{CEMA} illustrates percentile plots for heat release (top row) and CEMA
(middle row).  In the percentile plots, the lowest
blue curve and the highest red curve correspond to the 1 and
100 percentiles ($p_{0.01}$ and $p_{1}$), respectively. The blue curves
correspond to $p_{0.01 \ldots 0.1}$, the green curves to $p_{0.2 \ldots 0.9}$
and the red curves to $p_{0.91 \ldots 1}$. 
The \pmetric-indicator evaluated using \Eqn{pmetric} for $\alpha=0.94$,
$\beta=0.98$ and $\gamma=0.01$ is shown in the bottom row of \Fig{CEMA}. 
Results are generated for four test cases described in \Tab{usecases}.
The vertical dotted lines were identified  by a domain expert who, via examination
of heat release and CEMA percentile plots, visually
located the time steps in the simulation where the mesh resolution and I/O
frequency should  be increased. We refer to these time steps  
as the ``true'' trigger time steps we wish to identify with the 
\pmetric-indicator and trigger functions. Note, for the RCCI cases, there are two ignition ranges.
To simplify the following exposition, we focus on the second rise in the heat release rate
profiles, as this is the ignition stage of interest to the scientists. However,
we note that our approach is robust in identifying the first ignition stage as
well.

\begin{table}[h]
  \caption{\label{tab:usecases} Four Combustion Use Cases analyzed in this
study. The ``true'' trigger time ranges are estimated based on $95-100^{\textrm{th}}$
percentiles of the heat release rate. The computed time ranges were evaluated
using our quantile sampling approach.  For the RCCI cases, the trigger time
ranges are based on the High Temperature Heat Release, i.e., the second peak in
the Heat Release Rate profiles.}
\centering
 \begin{tabular}{|c|c|c|c|c|} 
 \hline
 Problem & Number of & Number of & ``True'' Trigger  & Computed  \\
 Instance & Grid Points & Species & Time Range & Trigger Time \\ \hline\hline
 HCCI, T=15   & 451,584     & 28  & 250-315 & 250-262\\  \hline
 HCCI, T=40   & 451,584     & 28  & 175-225 & 213-220 \\  \hline
 RCCI, case 6 & 2,560K   & 116 & 38-50   & 28-45\\  \hline
 RCCI, case 7 & 2,560K--10,240K   & 116 & 42-58   & 35-50\\  \hline
 \end{tabular}
 \end{table}

\subsection{Defining a Trigger}
\label{sec:trigger} 

In addition to defining a noise-resistant indicator function, we also need to define a trigger
function that returns a boolean value, capturing whether a property of the
indicator has been met. Looking at \Fig{CEMA}, we notice that across all
experiments from \Tab{usecases}, the \pmetric-indicator is decreasing during the true trigger time
step windows.  Therefore, we seek to find a value 
$\thresh_{\pmetric} \in (0,1)$, such that $\pmetric_{\alpha,\beta,\gamma}(t)$ crosses $\thresh_{\pmetric}$ \emph{from
above}, as the simulation time $t$ progresses. 

\Fig{Thresholds} plots the trigger time steps as a function of
$\thresh_{\pmetric}$ for a variety of configurations of $\alpha$ and $\beta$
for the four use cases desribed in \Tab{usecases}.  
The horizontal dashed lines indicate the true trigger range identified by our
domain expert.  We consider those values of $\thresh_{\pmetric}$ that fall within the
horizontal dashed lines to be viable $\thresh_{\pmetric}$ values for our trigger.  We find that across all
use cases, there are similar viable ranges of values for $\thresh_{\pmetric}$
where the predicted trigger time steps do not exhibit large variations. In
\Fig{Triggers_thresh} we provide a plot the trigger time steps as a function of
$\thresh_{\pmetric}$ for $\pmetric_{\alpha=0.94, \beta=0.98,\gamma=0.01}(t)$
that shows the viable range of $\thresh_{\pmetric}$ is $[ 0.725, 0.885 ]$.

\begin{figure}[htpb]
\centering
\includegraphics[width=0.8\textwidth]{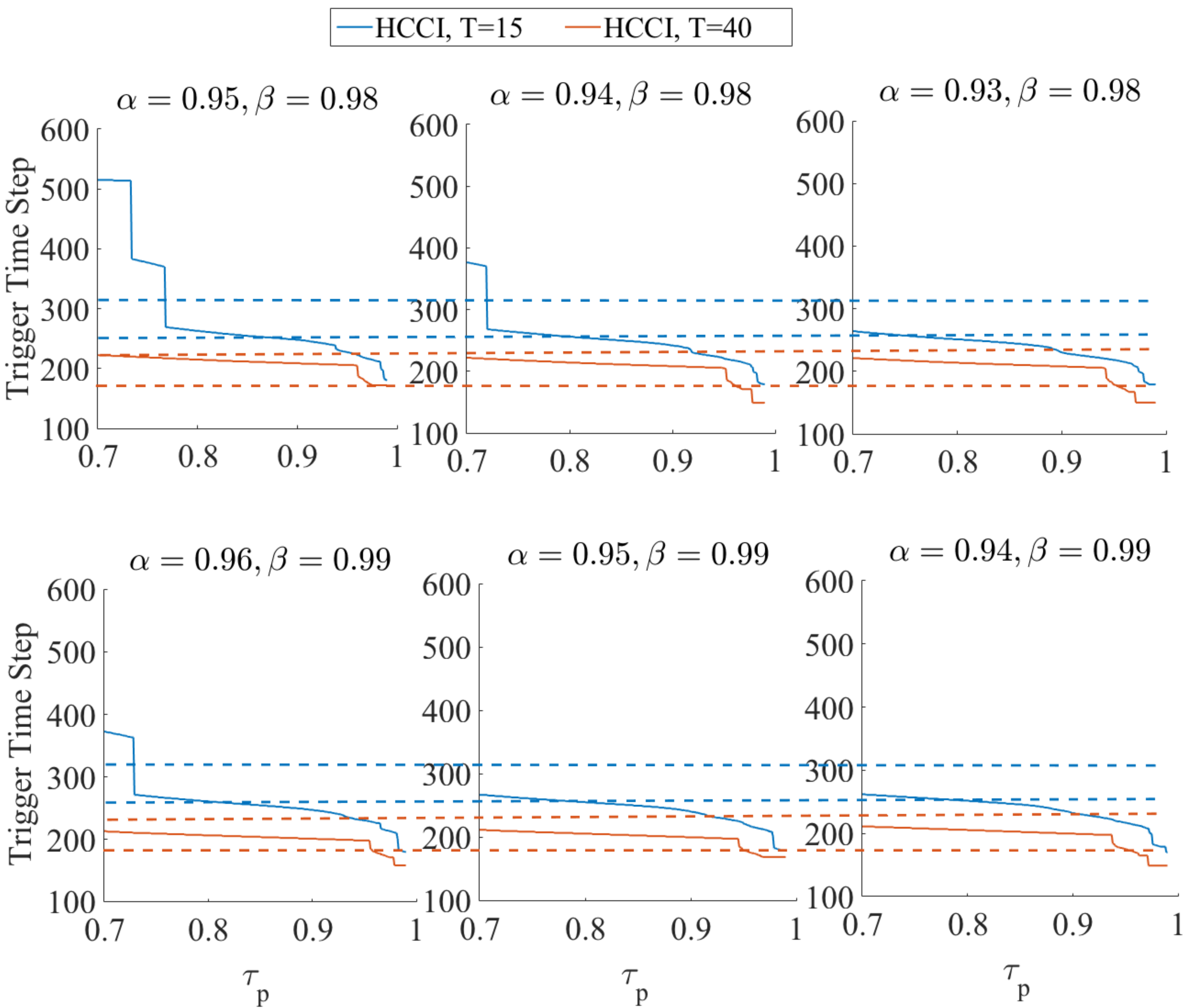}
\includegraphics[width=0.8\textwidth]{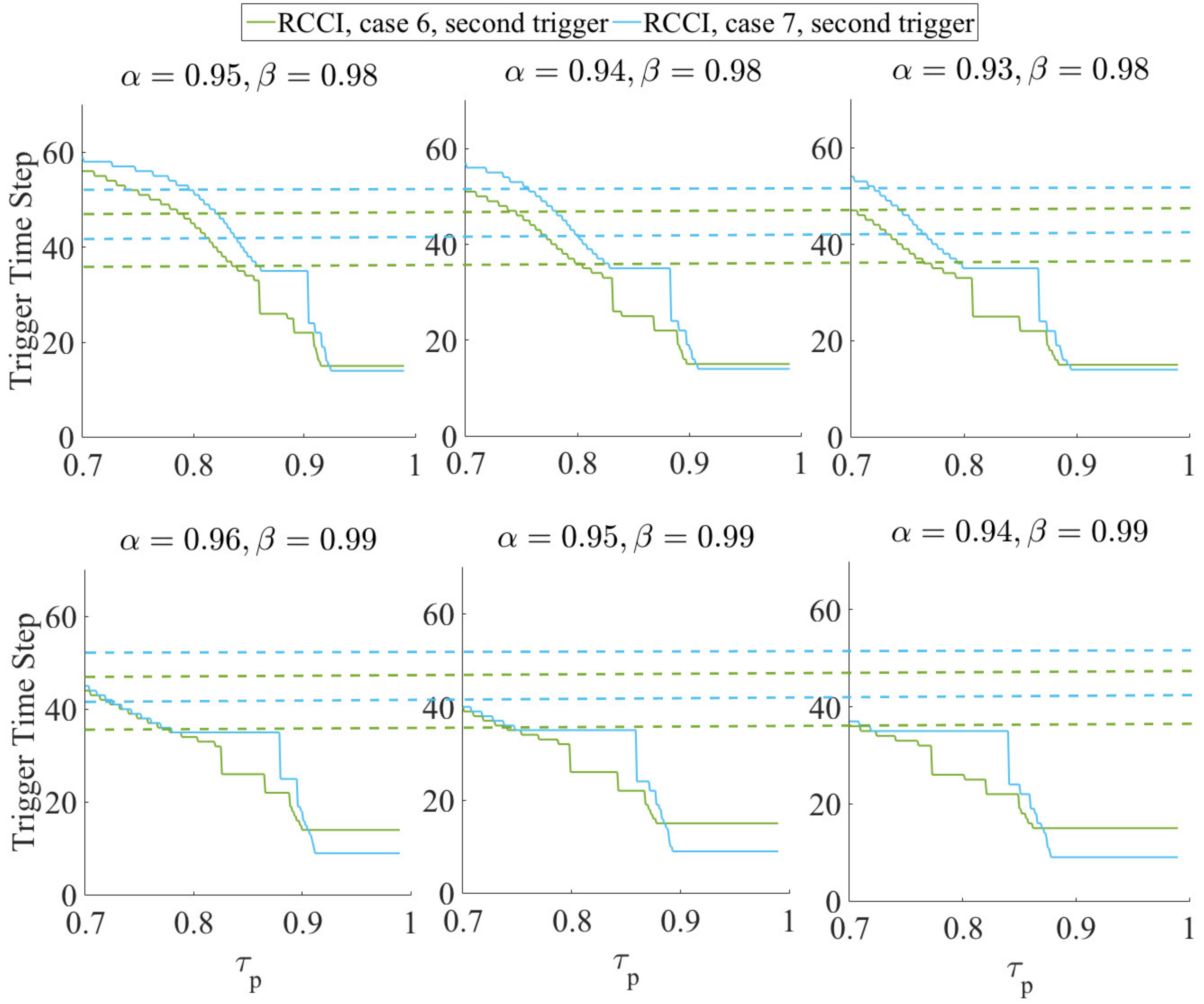}
\caption{\label{fig:Thresholds} 
This figure shows the trigger time steps as a function of
$\thresh_{\pmetric}$ for a variety of configurations of $\alpha$ and $\beta$
for the four use cases desribed in \Tab{usecases}.  
The horizontal dashed lines indicate the true trigger range identified by our
domain expert.  The \pmetric-indicator is evaluated with percentiles computed
using all $N$ grid points for the different use cases, as indicated. } 
\end{figure}

\begin{figure}[t]
\centering
\includegraphics[width=0.6\textwidth]{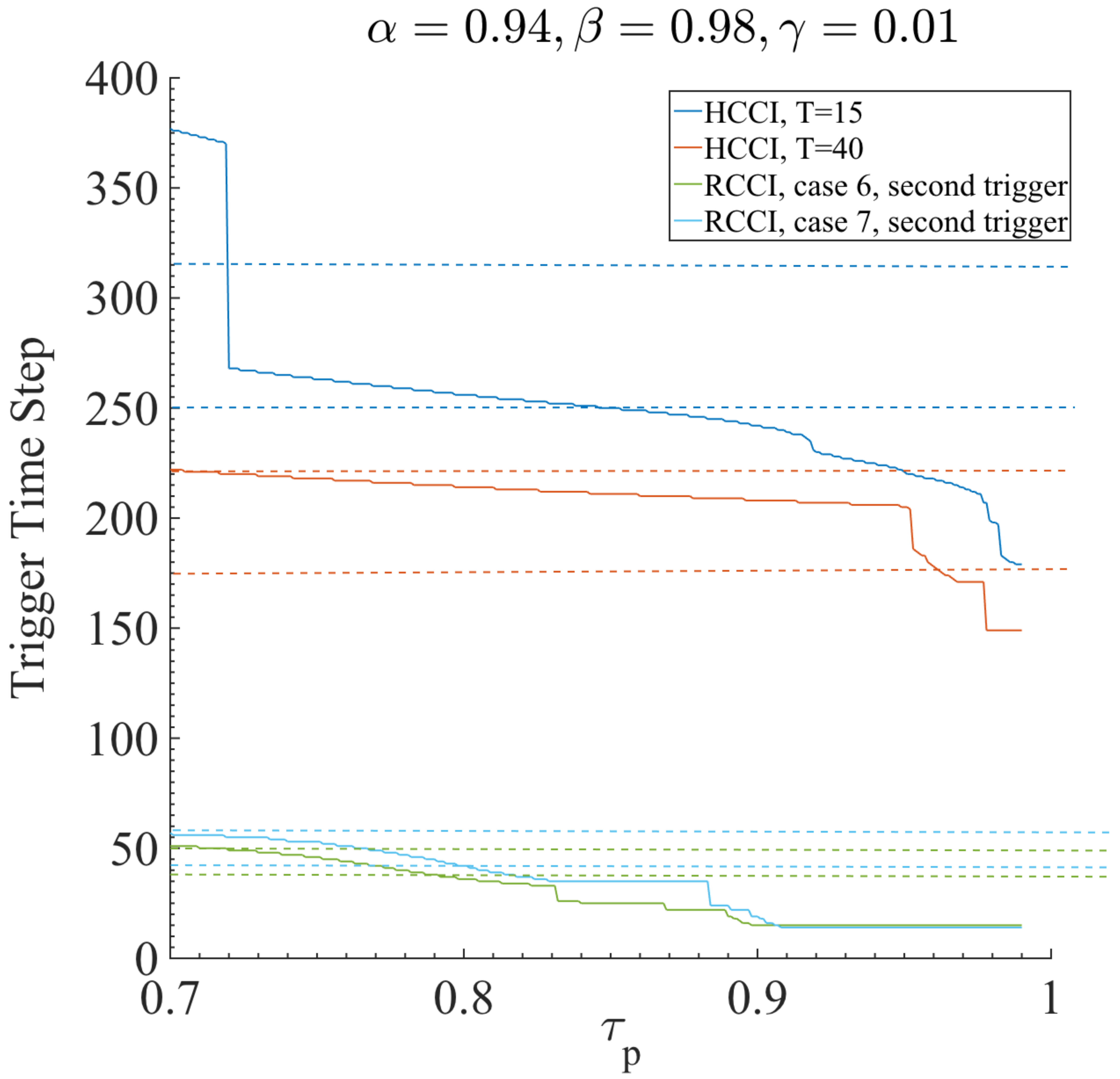}
\caption{\label{fig:Triggers_thresh} This figure plots the trigger time steps as a function of
$\thresh_{\pmetric}$ for $\pmetric_{\alpha=0.94, \beta=0.98,\gamma=0.01}(t)$.
  There is a range of viable values of $\thresh_{\pmetric} \in [0.725, 0.885]$ that
  predict early stage heat release.
}
\end{figure}
 
\section{Computing Indicators and Triggers Efficiently: A Sublinear Approach}
\label{sec:sample} 
The previous section showed that a CEMA-based \pmetric-indicator and
trigger are robust to noise fluctuations and act as a precursor to rapid heat
release in combustion simulations. In this section we provide some details
regarding its computational cost, which can be prohibitive for large-scale
simulations.  We then introduce an efficient method  to estimate the
\pmetric-indicator. Our method is based on quantile sampling and it comes with
provable bounds on accuracy as a function of the number of samples.  Most
importantly,  the required number of samples for a specified accuracy is
independent of  the size of the problem, hence our sampling based algorithms
offers  excellent scalability. 

\subsection{Computational Cost of CEMA}
Although CEMA is useful for predicting ignition, it is expensive to
compute and thus historically has not been used as a predictive measure. 
Computing CEMA values involves constructing  a large, dense matrix at every grid
point, and computing its eigendecomposition. For the use cases considered here, the time taken to compute the
CEMA values scales as the time taken to compute the eigendecomposition of an $M\times M$ matrix,
where $M$ is the number of species.
Since the Jacobian does not have any spatial structure, the time taken for the CEMA computation step is
$\textbf{O}(M^3)$, which makes it increasingly expensive for larger chemical mechanisms. As seen in 
\Tab{cema_costs}, for the ethanol
HCCI case presented here that was simulated with $28$-species, the cost of computing the CEMA value
at every grid point was approximately $5$ times the cost of one simulation time-step. The RCCI case, 
on the other hand included $116$-species and the cost was roughly $60$ times that of a single time step.
Although it is infeasible to compute CEMA at every grid point,  our indicator
function is defined in 
terms of distribution of  CEMA values, $\cema(t)$, and this can be easily approximated by a sampling mechanism, that has provable 
guarantees on the error. 

\subsection{Approximating percentiles by sampling} \label{sec:samp}
To remind our notation,  we have an array $A \in \RR^N$, in sorted order.
Our aim is to estimate the $\alpha$-percentile of $A$. (We use $\pt_\alpha$ to denote
the percentiles.) Note that this 
is exactly the entry $A_{\lceil \alpha N \rceil}$. Here is a simple sampling procedure. \\
\begin{compactenum}
	\item Sample $\samp$ independent, uniform indices $r_1, r_2, \ldots, r_\samp$
	in $\{1,2,\ldots,N\}$. \\ Denote by $\widehat{A}$ the sorted array $[A(r_1),A(r_2),\ldots,A(r_\samp)]$.
	\item Output the $\alpha$-percentile of $\widehat{A}$ as the estimate, $\widehat{\pt}_\alpha$. \\
\end{compactenum}

In the next section, we quantitatively show that our estimation, $\widehat{\pt}_\alpha$, is approximately close to the the true $\pt_\alpha$.
Such sampling arguments were also used in automatic generation of colormaps for massive data~\cite{TBSP13}.

\subsection{Theoretical bounds on performance}
Our analysis  relies on the  following fundamental result by Hoeffding, which provides a concentration inequality for sums of independent random variables.

\begin{theorem}[Hoeffding \cite{Ho63} or Theorem 1.1 in~\cite{DuPa09}]
  \label{thm:Hoeffding}
  Let $X_1, X_2, \dots, X_k$ be independent random variables with $0
  \leq X_i \leq 1$ for all $i=1,\dots,k$.  Define $X =
  \frac{1}{k} \sum_{i=1}^k X_i$. 
  For any positive $t$, we have \[ \Pr[|X - \EX[X]| \geq t] \leq 2 \exp(-2 t^2/k).\]
\end{theorem}

\begin{lemma} \label{lem:perc} Fix $\alpha$ and parameters $\delta, \eps \in (0,1)$, such that $\alpha > \eps$ and $\alpha + \eps \leq 1$.
Set $\samp = \cei{\frac{\log(4/\delta)}{2\eps^2}}$. Then $\widehat{\pt}_\alpha \in [\pt_{\alpha - \eps},\pt_{\alpha+\eps}]$
with probability at least $1-\delta$.
\end{lemma}

\vspace{2ex}
\begin{proof} Let $X_i$ be the (Bernoulli) indicator random variable for the event $r_i < (\alpha-\eps)N$. (i.e.,$X_i = 1$, if $r_i < (\alpha-\eps)N$,
and zero otherwise.) Note that 
\[\EX[X_i] = \Pr[X_i = 1]=\Pr[r_i < (\alpha-\eps)N] < \alpha - \eps.
\]
By linearity of expectation, $\EX[\sum_{i\leq \samp} X_i] < \samp(\alpha - \eps)$. By Hoeffding's inequality (\Thm{Hoeffding}),
\[\Pr[\sum_{i \leq \samp} X_i - \EX[\sum_{i \leq \samp} X_i] \geq \eps \samp] \leq 2 \exp(-2 \eps^2\samp). \]
The latter is at most $\delta/2$. Thus, with probability $>1 - \delta/2$,
\[\sum_{i \leq \samp} X_i \leq \EX[\sum_{i \leq \samp} X_i] + \eps k < \alpha k.\]
In plain English, with probability at least $1-\delta/2$, the number of random indices
strictly less than $(\alpha-\eps)N$ is strictly less than $\alpha k$.

We repeat a similar argument with indicator random variable $Y_i$ for the event $r_i < (\alpha+\eps)N$. So $\EX[Y_i] > \alpha + \eps$
and $\EX[\sum_{i \leq \samp} Y_i] > \samp(\alpha + \eps)$. By Hoeffding's inequality,
\[ \Pr[\EX[\sum_{i \leq \samp} Y_i] - \sum_{i \leq \samp} Y_i > \eps \samp] < \delta/2.
\]
With probability at least $1-\delta/2$, the number of random indices at most $(\alpha+\eps)N$
is strictly more than $\alpha k$.

By the union bound on probabilities, both events hold simultaneously with probability $>1-\delta$.
In this situation, the $\alpha$-percentile of the sample lies between the $\pt_{\alpha-\eps}$
and $\pt_{\alpha+\eps}$.
\end{proof}
\begin{table}[t]
\caption{ Additional cost associated with computing CEMA indices at every grid point (no sub-sampling) for two chemical mechanisms. Cost is given in seconds of wall-clock time per overall simulation time step. \label{tab:cema_costs}}
\begin{center}
 \begin{tabular}{| c | c | c | c | c |}
 \hline \hline
Fuel         & Mechanism       &  Cost without    & Cost with       &  Cost     \\
             & size (species)  &    CEMA (sec/ts) & CEMA (sec/ts)   & factor    \\ \hline \hline 
 Ethanol     &    $28$           &     $0.3$      &  $1.5$            &  $5$     \\
\hline
Primary Reference & & & & \\
Fuel (PRF) & $116$          &     $3.0$      &  $180.0$          &  $60$    \\\hline \hline 
\end{tabular}
\end{center}
\end{table}  

It bears emphasizing that the number of  samples,  $\samp$, is independent of  the problem size, $N$, and only depends on $\eps,\delta$.
So, the required number of samples  only depends on the desired accuracy, not on the size of the data. This is
the key to the scalability of our approach.

To compute the \pmetric-indicator, $\pmetric_{\alpha,\beta,\gamma}$, we just employ the procedure
above to get estimates $\widehat{\pt}_{\alpha}, \widehat{\pt}_{\beta}, \widehat{\pt}_{\gamma}$.
We can use the same samples (with only an additive increase to $\samp$) for all estimates, so we do not have to repeat the procedure $3$ times.
That yields the approximate \pmetric-indicator, denoted by $\widehat{\pmetric}_{\alpha,\beta,\gamma}$.

\begin{theorem} \label{thm:pmetric} Fix $\alpha, \beta, \gamma$ and parameters $\delta, \eps \in (0,1)$
such that $\alpha < \beta - 2\eps$ and $\eps < \min(\alpha,\beta,\gamma)$. Set $\samp = \cei{\frac{\log(12/\delta)}{2\eps^2}}$.
With probability $>1-\delta$, 
\[ \widehat{\pmetric}_{\alpha,\beta,\gamma} \in [\pmetric_{\alpha-\eps,\beta+\eps,\gamma+\eps},\pmetric_{\alpha+\eps,\beta-\eps,\gamma-\eps}].\]
\end{theorem}

\begin{proof} Apply \Lem{perc} for each of $\alpha, \beta, \gamma$ with error parameter $\delta/3$.
That gives the value of $\samp$ as stated in the theorem. By the union bound on probabilities,
all the following hold simultaneously with probability $>1-\delta$:
$\widehat{\pt}_{\alpha} \in [\pt_{\alpha-\eps},\pt_{\alpha+\eps}]$,
$\widehat{\pt}_{\beta} \in [\pt_{\beta-\eps},\pt_{\beta+\eps}]$, and
$\widehat{\pt}_{\gamma} \in [\pt_{\gamma-\eps},\pt_{\gamma+\eps}]$,
Hence,
$$ \widehat{\pmetric}_{\alpha,\beta,\gamma} = \frac{\widehat{\pt}_{\alpha} - \widehat{\pt}_{\gamma}}{\widehat{\pt}_\beta - \widehat{\pt}_\gamma}
\leq \frac{\pt_{\alpha+\eps} - \widehat{\pt}_{\gamma}}{\pt_{\beta-\eps} - \widehat{\pt}_\gamma}
\leq \frac{\pt_{\alpha+\eps} - \pt_{\gamma-\eps}}{\pt_{\beta-\eps} - \pt_{\gamma-\eps}} $$
For the last inequality, observe that for fixed $x < y$, $(x-z)/(y-z)$ is a decreasing function of $z$.
An analogous argument proves the lower bound for $\widehat{\pmetric}_{\alpha,\beta,\gamma}$.
\end{proof}

\begin{figure}[thbp] 
   \centering
   \includegraphics[width=3.5in]{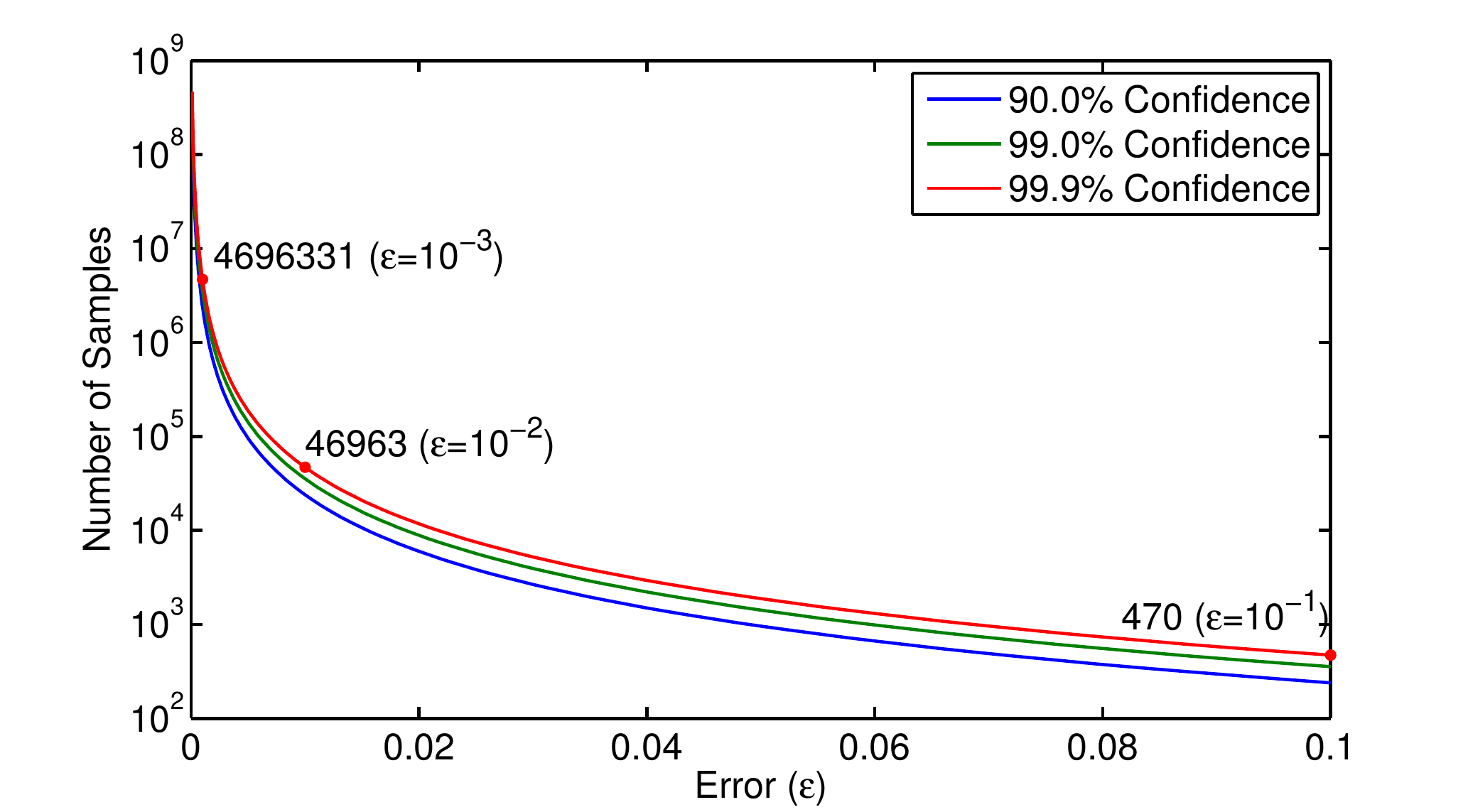} 
   \caption{ The number of samples needed for different error rates and
    different levels of confidence. A few data points at 99.9\%
    confidence are highlighted.}
   \label{fig:levelcurves}
\end{figure}

\Fig{levelcurves} shows the number of samples needed
for different error  and confidence rates. We show three different curves for
difference confidence levels. Increasing the confidence has minimal
impact on the number of samples. The number of samples is fairly low
for error rates of 0.1 or 0.01, but it increases with the inverse
square of the desired error. 
Nonetheless, the four million samples required for an error rate of $\eps = .001$ at 99.9\% confidence
requires only tens of samples per processor at the extreme scales.  In practice, $\eps = .01$ at 99.9\% confidence, thus 48,000  samples was enough to compute robust estimates. 

\subsubsection{Interpretation of the bounds}
Our bounds on  quantile estimations are based on which quantile we sample.  That
is our sampling algorithm may return the $\alpha+\eps$-th quantile  instead of
$\alpha$-th quantile, and we can quantify this error, $\eps$, as a function of
the number of samples.  However, it does not quantify  the difference between
${\pt}_{\alpha}$  and  $\widehat{\pt}_{\alpha}={\pt}_{\alpha+\eps}$.
Subsequently, \Thm{pmetric} shows that the range we sample can be made
arbitrarily close to the original range by increasing the number of samples,
and  bounds the error in the range for any given sample size. Yet, it does not bound the difference between $\pmetric_{\alpha,\beta,\gamma}$ and $\widehat{\pmetric}_{\alpha,\beta,\gamma}$, which depends on the distribution of the data.  However,  this is not a critical issue for our purposes, as we argue next, and empirically verify with experimental results in the next section.  

The difference between $\pmetric_{\alpha,\beta,\gamma}$ and $\widehat{\pmetric}_{\alpha,\beta,\gamma}$ will be disproportional to $\eps$ only when there are gaps in the distribution around one of the three parameters,  $\alpha$, $\beta$, or $\gamma$.  Note that these three parameters are user specified, and they are used to quantify the range of top percentiles. If the underlying distribution is such that we expect many such gaps frequently,  then  our  metric itself  will be extremely sensitive to the choice of  the input parameters, even if compute the metric exactly. That is our metric should not produce vastly different  results when we choose $\alpha =0.940$  or $\alpha =0.941$.  However, there is  still a possibility that such gaps  may form, just like any other  low probability event.  One trick to  improve robustness of our metric against such low probability events is to pick the  input parameters randomly from a specified range.  That is instead of specifying $\gamma$ as $0.03$ we can pick it randomly in the range $[0.02,0.04]$.  By such randomness, even if there is a gap at point, $\phi$, the probability that $\phi$ is in the $[\beta-\eps,\beta+\eps$] range will approach to zero with increasing number of  samples and thus decreasing $\eps$.    
  
However, we want to note that this is only a theoretical exercise. 
From a practical perspective, we have not observed  any  gaps in the distribution and ${\pt}_{\alpha+\eps}$ is a good estimate for ${\pt}_{\alpha}$, and it gets better with more samples.  For the  experiments in the remainder of this  paper, we have not  used  the randomization technique when choosing the parameters, $\gamma$, $\alpha$, and $\beta$.  

\subsection{Empirical  evaluation of the sampling based  algorithm}
In this section we present our empirical evaluation of the proposed  sampling
techniques.  Experiments  in this section  will focus on only the  evaluation
of  the proposed algorithm,  since we first want to verify  that the proposed
sampling  technique accurately estimates the \pmetric-indicator. In the next
section, we will put all pieces together  and study how  the
\pmetric-indicator and  the proposed technique perform together \emph{in situ}
as the simulation is running.

In the first set of experiments, we investigate the error in quantile ranges. For these experiments, we  use 16 randomly selected instances  of CEMA  distributions from various HCCI and RCCI simulations.   These  instances are named  such that the first part refers to the simulation  type  and case, and the last part refers to the  time step. 
We use sampling to estimate the $\alpha=0.94$ percentile, which returns an entry from  the distribution. Then we check the  percentile of this entry in the full data, say $\alpha+\eps$.  In the first set of experiments, we focus on this difference $\eps$, which we bounded in our theoretical analysis.   

\Fig{quantiles} (a) presents  the results of our investigation into the error of
quantile ranges for various data sets and increasing number of samples: 12,000, 24,000, and 48,000.  For this figure, we ran our sampling  algorithm 100 times  for each instance (i.e.,  a data set  and number of samples combination) and  computed $\eps$.  The figure presents the average  $|\eps |$ for each instance. As the figure shows, sampling  yields accurate  estimations in all data sets, and the error drops with increasing number of samples. It becomes extremely small for 48,000 samples.   Here an error of 0.001 means we will be using ${\pt}_{0.941}$ quantile instead of ${\pt}_{0.940}$.  We  did not find it necessary to investigate increasing the number of samples further.

\begin{figure}[htbp] 
   \centering
   \vspace{-1ex}
   \begin{tabular}{cc}
   \hspace*{-0.08\textwidth}  
     \includegraphics[height=1.5in]{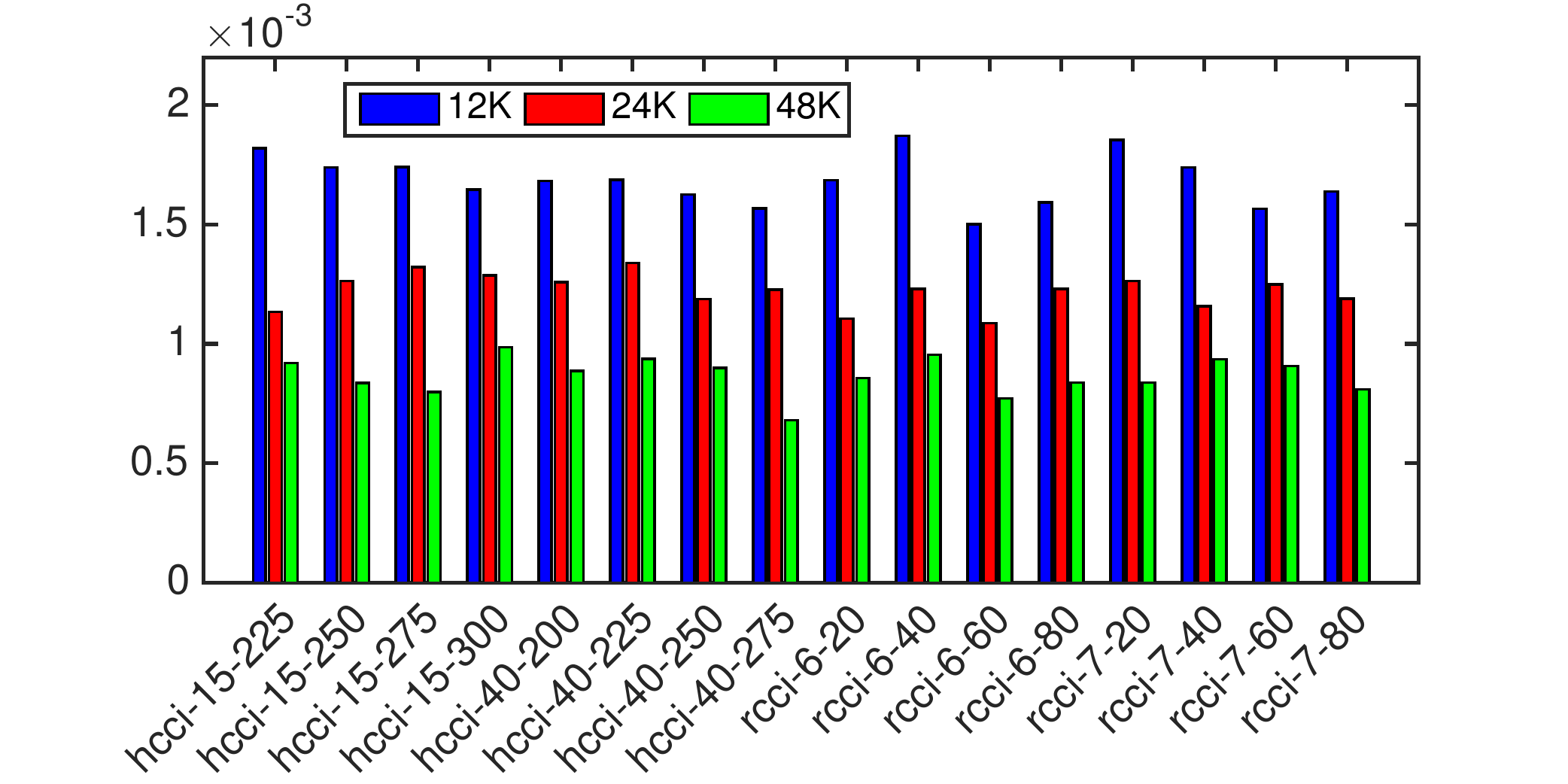} &
     \includegraphics[height=1.5in] {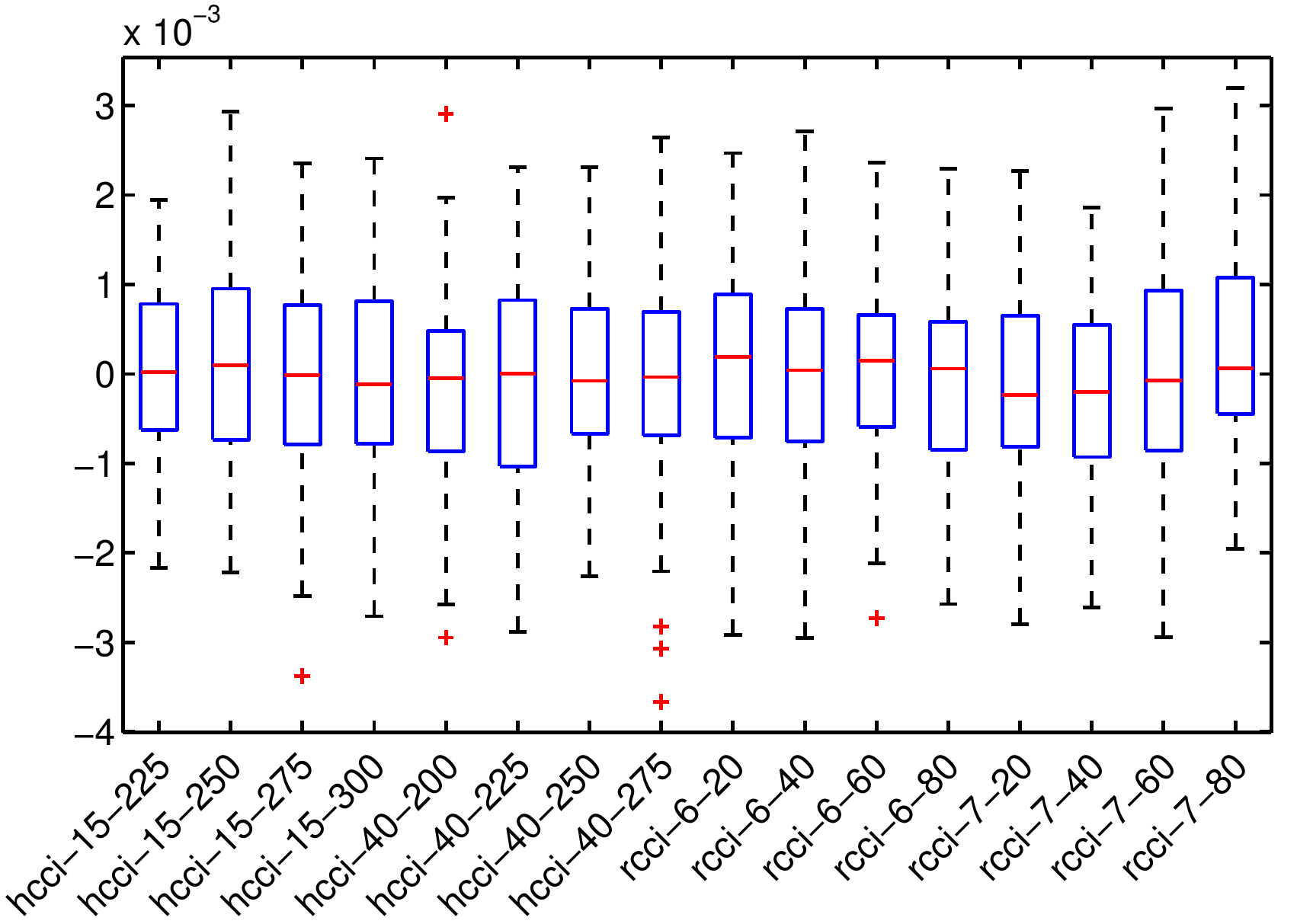} \\
     \textrm{(a)}Average error  & \textrm{(b)} Distribution of errors\\
   \end{tabular}
   \caption{(a) Error in the percentile sampled as
   the average of  absolute values of 100 runs on various data sets  with
 increasing number of samples. (b)  The distribution of errors  of each data set for the percentile study using 48,000 samples. }
   \label{fig:quantiles}
\end{figure}

In \Fig{quantiles} (b), we show the  distribution of errors in percentiles
sampled for the 100 runs of each dataset using 48,000 samples. 
In this  plot, the central mark (red) shows the median error, while the edges of the (blue) box are the 25th
and 75th percentiles. The whiskers extend to the most extreme points considered not
to be outliers, and the outliers (red plus marks) are plotted individually. As  this figure shows,  the estimates are consistently accurate, and  the results in practice are  much better than  those indicated by \Thm{pmetric}. According this theorem, 48,000 samples  lead to  an error of $\approx 0.01$ with a confidence of \%99.  This means in 100  experiments we expect to have 1 run for which the error is  $>0.01$.  However, in the 16  data sets with 100 runs each the maximum error was $0.005$, half of what the upper bound indicates. These results show that  sampling  enables us to sample a quantile that is consistently accurate.

In the next set of experiments, we look at how our indicator is affected by the
minor errors in the quantile. More specifically, we want to see  how the
difference between ${\pt}_{\alpha}$ and  ${\pt}_{\alpha+\eps}$ affect our
indicator.  
For those experiments,  we  used 16 randomly selected instances  of
CEMA  distributions from various HCCI and RCCI simulations and computed  the
\pmetric-indicator exactly, and using sampling  with parameters, $\gamma=0.01$ $\alpha=0.94$, and $\beta=0.98$. 
\begin{figure}[th] 
   \centering
   \begin{tabular}{cc}
  \hspace*{-0.08\textwidth}    
  \includegraphics[height=1.6in]{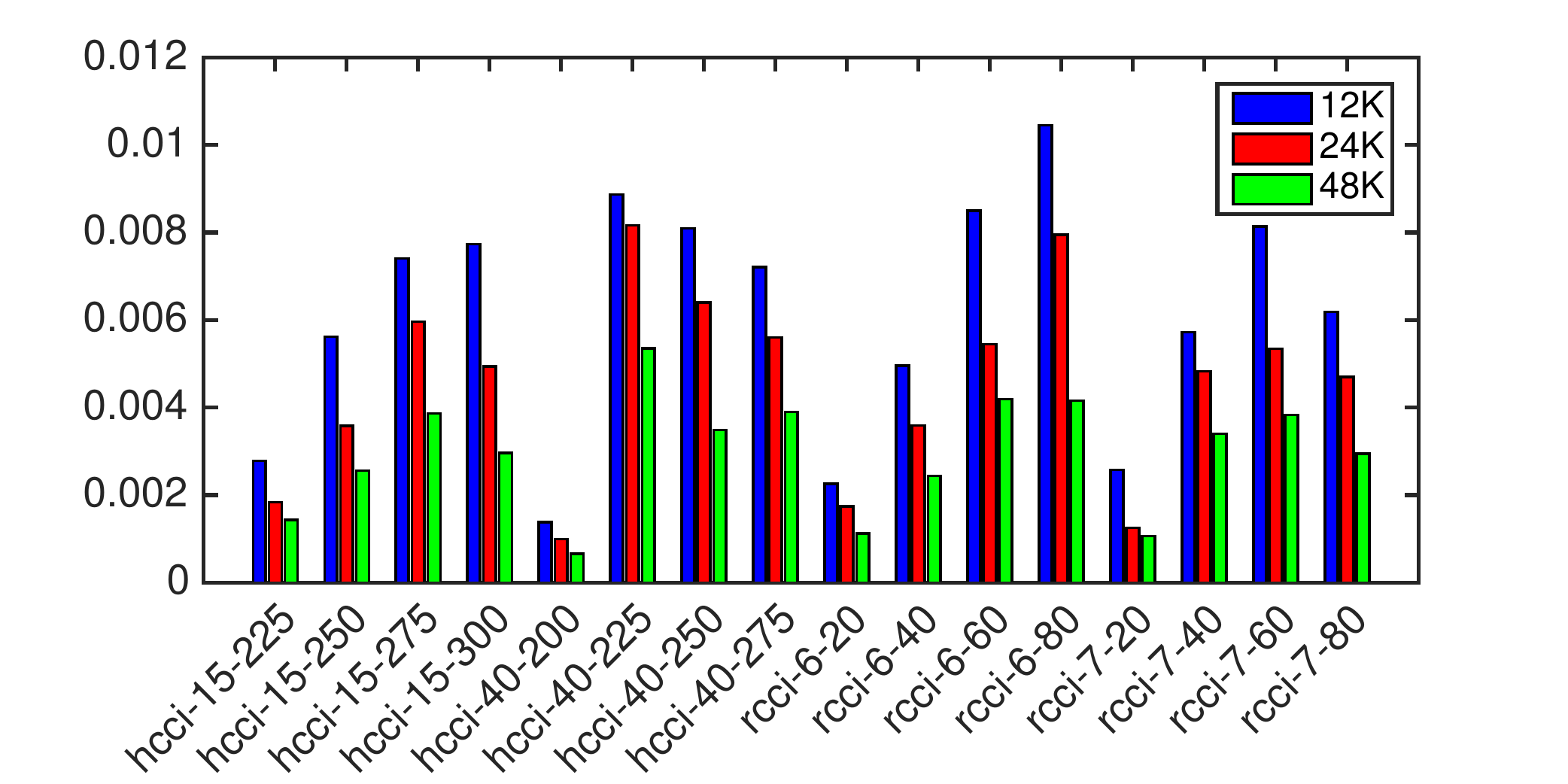}   &
  \includegraphics[height=1.6in]{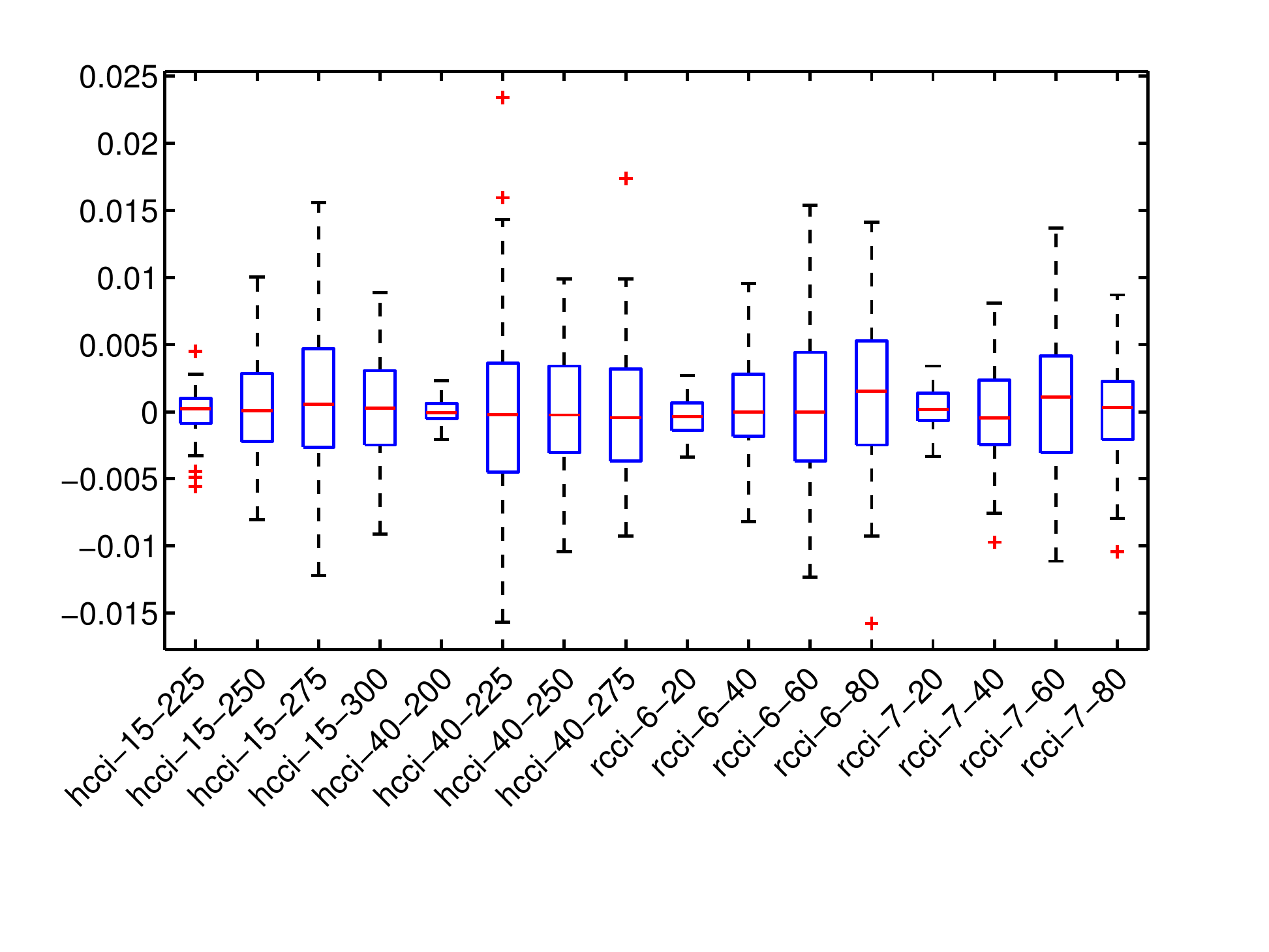} \\
 (a)  Average Error & (b) Distribution of errors  
\end{tabular}
\caption{(a)  Errors in  estimation of  the \pmetric-indicator for
 increasing number of samples.
(b)  The distribution of errors in  estimation of  the \pmetric-indicator  for 48,000
samples.}
   \label{fig:Boxquantiles}
\end{figure}
While we repeated the same experiments with different parameters, we have not observed any sensitivities of our tests to  the choice the parameters,  hence we will present results for only this setting.  

\Fig{Boxquantiles}(a) presents results for our indicator tests for various data sets and increasing number of
samples: 12,000, 24,000, and 48,000.  For this figure, we ran our sampling
algorithm 100 times  on each data set  and number of samples, and  computed
the difference between  exact  and estimated values of the \pmetric-indicator  
for each instance.  The figure presents the average absolute error for each 
instance. As the figure shows, sampling  produces accurate estimations in all 
data sets, and the error drops with increasing number of samples. The bounds 
on \Thm{pmetric} does not apply in this case. but regardless, the errors are 
very small, and certainly sufficient  to detect any trend in the distribution of the underlying values. 

\Fig{Boxquantiles}(b) shows the  distribution of errors for the
\pmetric-indicator test for 100 runs of each dataset using 48,000 samples. 
In this  plot, the central mark (red) shows the median error, while the edges of the (blue) box are the 25th
and 75th percentiles. The whiskers extend to the most extreme points considered not
to be outliers, and the outliers (red plus marks) are plotted individually. As  this figure shows,  the estimates are consistently accurate. 

\begin{figure}[htpb]
\includegraphics[width=\textwidth]{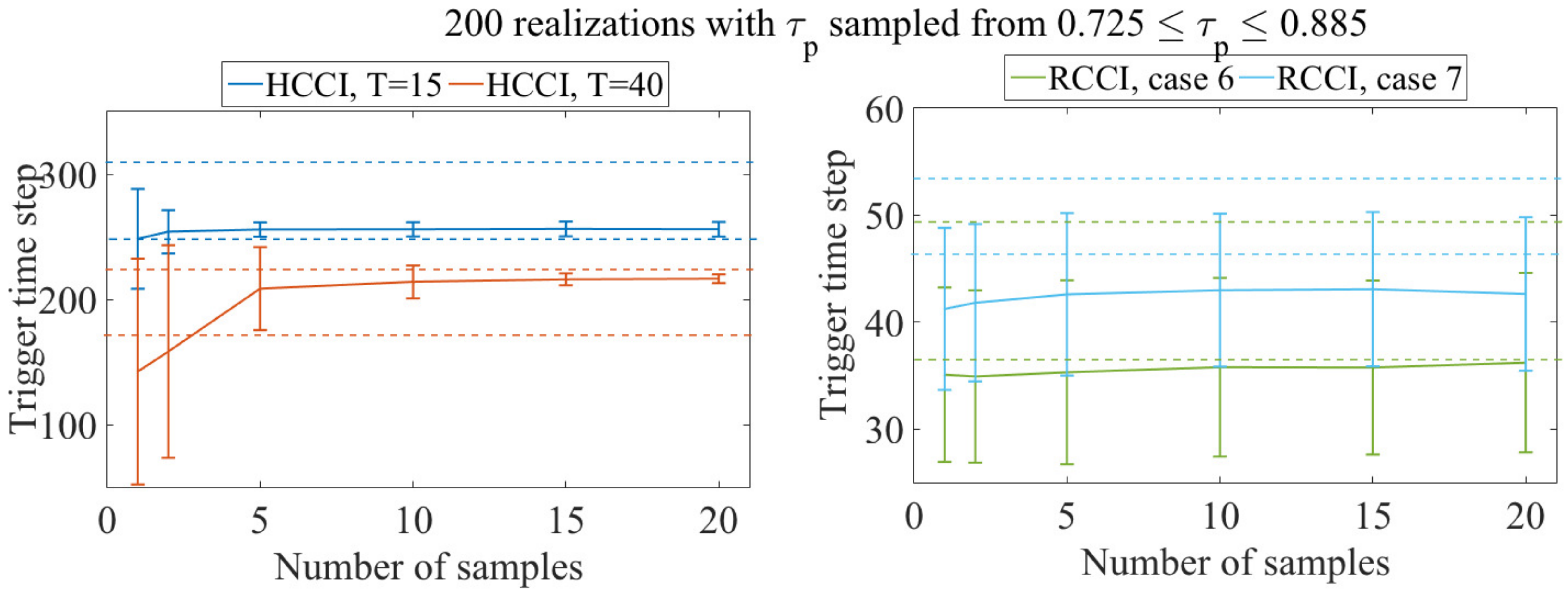} \\ \\
\caption{\label{fig:Triggers} Plots illustrating the variability of the trigger
time steps 
predicted by the \pmetric-indicator and trigger as a function of the number of
samples per processors.  The data for these plots was generated via 200
realizations of the \pmetric-indicator with $\alpha=0.94$, $\beta=0.98$ and
$\gamma=0.01$, and with $\thresh_{\pmetric}$ drawn from $[0.725, 0.885]$. 
 The horizontal dashed lines define the range of time steps within
which we would like to make the workflow transition (as identified by a domain expert).
}  
\end{figure}

Lastly, we performed a series of experiments examining the variation in the trigger time steps 
as a function of the number of samples used per processor.  The data for
\Fig{Triggers} was generated via 200 realizations of  the \pmetric-indicator with $\alpha=0.94$, $\beta=0.98$ and
$\gamma=0.01$, and $\thresh_{\pmetric}$ drawn from $[0.725, 0.885]$. The
horizontal dashed lines in this figure define the range of true trigger time steps within
which we would like to make the workflow transition (as identified by a domain expert).
This plot demonstrates that, even across a wide range of $\thresh_{\pmetric}$
values,  with a small number of samples per processor, the quantile sampling
approach can accurately estimate the  true trigger time steps as defined by the
domain expert. 
The next step will be  putting all
the pieces  together  to see how we can predict heat release \emph{in situ},
within a simulation run. 

\section{Putting all pieces together: Diagnosing heat release  with sublinear algorithms in S3D}
\label{sec:app}
  
The algorithm described in the previous section was deployed \emph{in situ} for a two-dimensional 
direct numerical simulation (DNS) of the ethanol HCCI problem (Case HCCI, T=40
in \Tab{usecases}).  
The DNS was run on half a million grid points with $784$ processors. $20$ points
were sampled at random 
from each processor for the CEMA analysis, generating a total of $15680$
samples for computing the trigger with the \pmetric-indicator. 
The parameters for computing the metric were chosen as follows: $\alpha=0.94$, $\beta=0.98$ and $\gamma=0.01$.
This corresponds to less than $4\%$ of the total simulation volume. 

\Fig{M_p_insitu}  shows the \pmetric-indicator being computed \emph{in situ} inside the simulation code. 
From top to bottom, the rows show the result when the indicator is computed every $10$,
$100$ and $1000$ time steps.  As the frequency of computing these indicator increases, 
the signal tends to get noisier. However, the overall trends and triggers do not change.   
These images show that our quantile-sampling approach provides a well defined trigger, using $\thresh_{\pmetric} \in
[0.725, 0.885]$ and can be used with confidence to predict the rapid rise in the heat release rate that we require to 
guide temporal and spatial refinement decisions \emph{in situ}. 

As can be inferred from \Tab{cema_costs}, performing the 
CEMA analysis on all grid points would increase the cost of the simulation by a factor of $5$, or $400\%$, which 
is clearly infeasible. Using the sublinear sampling algorithm on the other hand, incurs an overhead of only $1\%$ 
on the total simulation cost when performed every ten time-steps.  The cost
savings are even more dramatic in larger, three-dimensional production runs, as
the number of samples required does not increase with the number of grid
points.  Furthermore, we note the cost savings are further increased for 
larger mechanisms such a primary reference fuel (PRF), composed of a blend of
iso-octane and n-heptane.  For the PRF mechanism, the CEMA 
overhead without sublinear sampling would be a factor of $60$ or more. We plan
to deploy this algorithm \emph{in situ} in future 
large three-dimensional production simulations using S3D, especially with large chemical mechanisms such as PRF.     

\begin{figure}[htpb]
\includegraphics[width=1.0\textwidth]{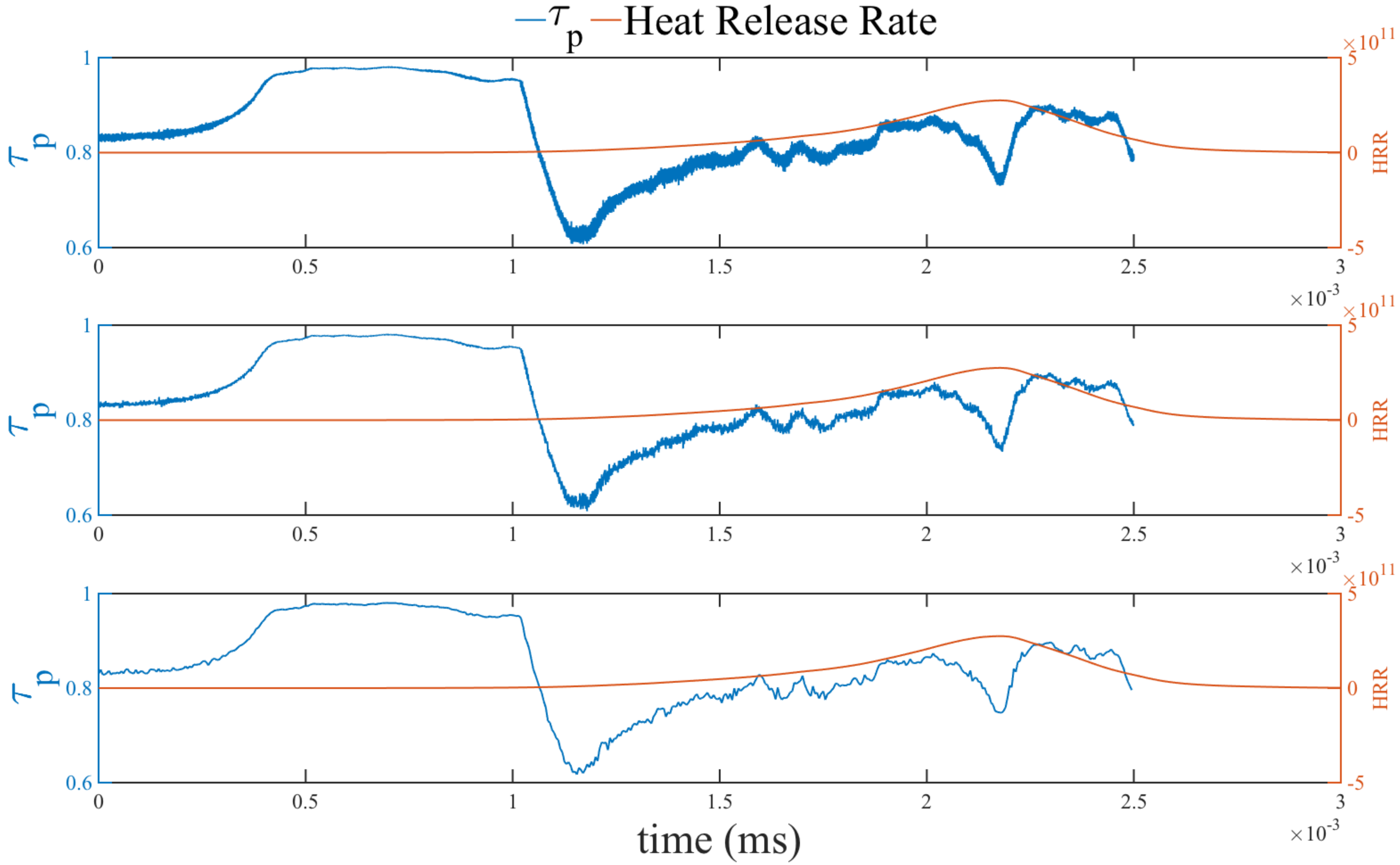}
\caption{\label{fig:M_p_insitu} Plots showing \pmetric-indicator being computed every $10$ (top), 
$100$ (middle) and $1000$ (bottom) time steps. }
\end{figure}

\section{Conclusion}
\label{sec:conc} 
We have proposed an approach for enabling dynamic, adaptive, extreme-scale 
scientific computing workflows.  We introduce the notion of indicators and 
triggers that are computed \emph{in situ}, that support data-driven control-flow 
decisions based on the simulation state. For those indicators and triggers that
are computationally prohibitive to compute, we demonstrate how sublinear 
algorithms enable their estimation with high confidence while incurring 
extremely low computational overheads.

Throughout this paper, we demonstrate our approach in practice using a proposed 
indicator to detect changes in the underlying physics of a combustion simulation.  
The goal of the indicator is to predict rapid heat release in direct numerical simulations of turbulent 
combustion. We show that chemical explosive mode analysis (CEMA) 
can be used to devise a noise-tolerant indicator for rapid increase in heat release.
Specifically,  we study the distribution of  CEMA values, and show that
heat release is preceded by a decrease in the range of  top quantiles in this
distribution.  We devise an indicator to quantify this intuition and show
that it can serve as a robust precursor for heat release. 
The cost of  exhaustive computation of CEMA values dominates  
the total simulation time, and we propose a sublinear algorithm based on
quantile sampling to overcome this computational bottleneck.  Our algorithm  
comes with provable error/confidence
bounds, as a function of the number of samples. Most importantly,  the number
of samples is independent of the problem size, thus our proposed sampling
algorithm offers perfect scalability.  Our experiments show that our
sublinear algorithm is  nearly as  accurate as  an exact algorithm  that relies
on exhaustive computation, while  taking only a fraction of the time.
Essentially,  sampling in this case  provides the algorithmic foundation to
turn a critical yet intractable indicator into a practical indicator that takes negligible time.   
Our experiments on homogeneous charge compression ignition (HCCI)  and reactivity controlled 
compression ignition (RCCI) simulations show that the proposed method can  predict heat release, 
and its computational overhead is negligible. 

The  impact of this paper is two fold. From the applications' perspective,
we have introduced an important  tool that enables adaptive workflows in 
combustion simulations, an important area of computational science and engineering.  
Our proposed methods, enable the controlling of mesh granularities, and adaptive I/O 
frequencies. It is already becoming critically important to have adaptive
control over these quantities, but will be crucial as we look
ahead to exascale computing. From an algorithmic perspective, our work showcased how sublinear 
algorithms, a recent development in  theoretical computer science can be
applied \emph{in situ}.  We believe these  algorithmic techniques hold great potential for 
\emph{in situ} analysis, and we expect them to be more widely used in the near future.


\begin{thebibliography}{10}

\bibitem{JITStaging}
{\sc H.~Abbasi, G.~Eisenhauer, M.~Wolf, K.~Schwan, and S.~Klasky}, {\em {Just
  In Time: Adding Value to The IO Pipelines of High Performance Applications
  with JITStaging}}, in Proc. of 20th International Symposium on High
  Performance Distributed Computing (HPDC'11), June 2011.

\bibitem{dav_exascale}
{\sc S.~Ahern, A.~Shoshani, K.-L. Ma, A.~Choudhary, T.~Critchlow, S.~Klasky,
  V.~Pascucci, J.~Ahrens, E.~W. Bethel, H.~Childs, J.~Huang, K.~Joy, Q.~Koziol,
  G.~Lofstead, J.~S. Meredith, K.~Moreland, G.~Ostrouchov, M.~Papka,
  V.~Vishwanath, M.~Wolf, N.~Wright, and K.~Wu}, {\em Scientific Discovery at
  the Exascale, a Report from the DOE ASCR 2011 Workshop on Exascale Data
  Management, Analysis, and Visualization}, 2011.

\bibitem{BaKo15}
{\sc G.~Ballard, T.~G. Kolda, A.~Pinar, and C.~Seshadhri},
  \href{http://arxiv.org/abs/1506.03872}{{\em Diamond sampling for approximate
  maximum all-pairs dot-product (mad) search}}, Tech. Report Arxiv:1506.03872,
  2015.

\bibitem{Bennett:2012}
{\sc J.~C. Bennett, H.~Abbasi, P.-T. Bremer, R.~Grout, A.~Gyulassy, T.~Jin,
  S.~Klasky, H.~Kolla, M.~Parashar, V.~Pascucci, P.~Pebay, D.~Thompson, H.~Yu,
  F.~Zhang, and J.~Chen},
  \href{http://conferences.computer.org/sc/2012/papers/1000a089.pdf}{{\em
  Combining in-situ and in-transit processing to enable extreme-scale
  scientific analysis}}, in {SC} '12: Proceedings of the International
  Conference on High Performance Computing, Networking, Storage and Analysis,
  Salt Lake Convention Center, Salt Lake City, {UT}, {USA}, November 10--16,
  2012, J.~Hollingsworth, ed., pub-IEEE:adr, 2012, IEEE Computer Society Press,
  pp.~49:1--49:9.

\bibitem{bhagatwala1}
{\sc A.~Bhagatwala, J.~H. Chen, and T.~Lu}, {\em Direct numerical simulations
  of {SACI}/{HCCI} with ethanol}, Comb. Flame, 161 (2014), pp.~1826--1841.

\bibitem{bhagatwala2}
{\sc A.~Bhagatwala, R.~Sankaran, S.~Kokjohn, and J.~H. Chen}, {\em Numerical
  investigation of spontaneous flame propagation under {RCCI} conditions},
  Comb. Flame,  (Under review).

\bibitem{visit:2011}
{\sc J.-M.~F. Brad~Whitlock and J.~S. Meredith}, {\em {Parallel In Situ
  Coupling of Simulation with a Fully Featured Visualization System}}, in Proc.
  of 11th Eurographics Symposium on Parallel Graphics and Visualization
  (EGPGV'11), April 2011.

\bibitem{chen09}
{\sc J.~H. Chen, A.~Choudhary, B.~de~Supinski, M.~DeVries, E.~R. Hawkes,
  S.~Klasky, W.~K. Liao, K.~L. Ma, J.~Mellor-Crummey, N.~Podhorski,
  R.~Sankaran, S.~Shende, and C.~S. Yoo}, {\em Terascale direct numerical
  simulations of turbulent combustion using {S}3{D}}, Computational Science and
  Discovery, 2 (2009), pp.~1--31.

\bibitem{DuPa09}
{\sc D.~Dubhashi and A.~Panconesi}, {\em Concentration of Measure for the
  Analysis of Randomized Algorithms}, Cambridge University Press, 2009.

\bibitem{paraview:ldav11}
{\sc N.~Fabian, K.~Moreland, D.~Thompson, A.~Bauer, P.~Marion, B.~Gevecik,
  M.~Rasquin, and K.~Jansen},
  \href{http://dx.doi.org/10.1109/LDAV.2011.6092322}{{\em The paraview
  coprocessing library: A scalable, general purpose in situ visualization
  library}}, in Proc. of IEEE Symposium on Large Data Analysis and
  Visualization (LDAV), October 2011, pp.~89 --96.

\bibitem{FischerSurvey}
{\sc E.~Fischer}, {\em The art of uninformed decisions: A primer to property
  testing}, Bulletin of EATCS, 75 (2001), pp.~97--126.

\bibitem{Ho63}
{\sc W.~Hoeffding}, \href{http://www.jstor.org/stable/2282952}{{\em Probability
  inequalities for sums of bounded random variables}}, Journal of the American
  Statistical Association, 58 (1963), pp.~13--30.

\bibitem{JhSePi15}
{\sc M.~Jha, C.~Seshadhri, and A.~Pinar},
  \href{http://dx.doi.org/10.1145/2736277.2741101}{{\em Path sampling: A fast
  and provable method for estimating 4-vertex subgraph counts}}, in Proceedings
  of the 24th International Conference on World Wide Web, WWW '15, Republic and
  Canton of Geneva, Switzerland, 2015, International World Wide Web Conferences
  Steering Committee, pp.~495--505.

\bibitem{JhSePi13}
\leavevmode\vrule height 2pt depth -1.6pt width 23pt,
  \href{http://dx.doi.org/10.1145/2700395}{{\em A space-efficient streaming
  algorithm for estimating transitivity and triangle counts using the birthday
  paradox}}, ACM Trans. Knowl. Discov. Data, 9 (2015), pp.~15:1--15:21.

\bibitem{kokjohn}
{\sc S.~L. Kokjohn, R.~M. Hanson, D.~A. Splitter, and R.~D. Reitz}, {\em Fuel
  reactivity controlled compression ignition (rcci): A pathway to controlled
  high-efficiency clean combustion}, Int. J. Engine. Res., 12 (2011), pp.~209--226.

\bibitem{KoPiPlSe13}
{\sc T.~G. Kolda, A.~Pinar, T.~Plantenga, C.~Seshadhri, and C.~Task},
  \href{http://dx.doi.org/10.1137/13090729X}{{\em Counting triangles in massive
  graphs with mapreduce}}, SIAM Journal on Scientific Computing, 36 (2014),
  pp.~S48--S77.

\bibitem{lu}
{\sc T.~Lu, C.~S. Yoo, J.~H. Chen, and C.~K. Law}, {\em Three-dimensional
  direct numerical simulation of a turbulent lifted hydrogen flame in heated
  coflow: a chemical explosive mode analysis}, J. Fluid Mech., 652 (2010),
  pp.~45--64.

\bibitem{modelPaper}
{\sc K.~Myers, E.~Lawrence, M.~Fugate, J.~Woodring, J.~Wendelberger, and
  J.~Ahrens}, {\em An in situ approach for approximating complex computer
  simulations and identifying important time steps}, tech. report, 2014.
\newblock arXiv:1409.0909v1.

\bibitem{nouanesengsy2014adr}
{\sc B.~Nouanesengsy, J.~Woodring, J.~Patchett, K.~Myers, and J.~Ahrens}, {\em
  {ADR} visualization: A generalized framework for ranking large-scale
  scientific data using analysis-driven refinement}, in Large Data Analysis and
  Visualization (LDAV), 2014 IEEE 4th Symposium on, IEEE, 2014, pp.~43--50.

\bibitem{RonSurvey}
{\sc D.~Ron}, {\em Algorithmic and analysis techniques in property testing},
  Foundations and Trends in Theoretical Computer Science, 5 (2009),
  pp.~73--205.

\bibitem{RubSurvey}
{\sc R.~Rubinfeld}, {\em Sublinear time algorithms}, International Conference
  of Mathematicians,  (2006).

\bibitem{RS96}
{\sc R.~Rubinfeld and M.~Sudan}, {\em Robust characterization of polynomials
  with applications to program testing}, SIAM Journal of Computing, 25 (1996),
  pp.~647--668.

\bibitem{SePiKo13}
{\sc C.~Seshadhri, A.~Pinar, and T.~G. Kolda},
  \href{http://dx.doi.org/10.1137/1.9781611972832.2}{{\em Triadic measures on
  graphs: The power of wedge sampling}}, in Proceedings of the 2013 SIAM
  International Conference on Data Mining, 2013, pp.~10--18.

\bibitem{SePiKo14}
{\sc C.~Seshadhri, A.~Pinar, and T.~G. Kolda},
  \href{http://dx.doi.org/10.1002/sam.11224}{{\em Wedge sampling for computing
  clustering coefficients and triangle counts on large graphs?}}, Statistical
  Analysis and Data Mining, 7 (2014), pp.~294--307.

\bibitem{shan}
{\sc R.~Shan, C.~S. Yoo, J.~H. Chen, and T.~Lu}, {\em Computational diagnostics
  for n-heptane flames with chemical explosive mode analysis}, Comb. Flame, 159
  (2012), pp.~3119--3127.

\bibitem{doe_arch}
{\sc R.~Stevens, A.~White, S.~Dosanjh, A.~Geist, B.~Gorda, K.~Yelick,
  J.~Morrison, H.~Simon, J.~Shalf, J.~Nichols, and M.~Seager}, {\em
  Architectures and technology for extreme scale computing}, tech. report,
  2009.

\bibitem{TBSP13}
{\sc D.~Thompson, J.~Bennett, C.~Seshadhri, and A.~Pinar},
  \href{http://dx.doi.org/10.1109/LDAV.2013.6675161}{{\em A provably-robust
  sampling method for generating colormaps of large data}}, in Large-Scale Data
  Analysis and Visualization (LDAV), 2013 IEEE Symposium on, Oct 2013,
  pp.~77--84.

\bibitem{glean:ldav11}
{\sc V.~Vishwanath, M.~Hereld, and M.~Papka},
  \href{http://dx.doi.org/10.1109/LDAV.2011.6092178}{{\em Toward
  simulation-time data analysis and i/o acceleration on leadership-class
  systems}}, in Proc. of IEEE Symposium on Large Data Analysis and
  Visualization (LDAV), October 2011.

\bibitem{Yu2010}
{\sc H.~Yu, C.~Wang, R.~W. Grout, J.~H. Chen, and K.-L. Ma}, {\em In-situ
  visualization for large-scale combustion simulations}, IEEE Computer Graphics
  and Applications, 30 (2010), pp.~45--57.

\end{thebibliography}
\end{document}